\documentclass[11pt,a4paper,leqno]{amsart}

\usepackage{amsmath,amsfonts,amssymb,amsthm}
\usepackage{color}
\usepackage{hyperref, url}
\usepackage{graphicx}
\usepackage{enumerate}
\usepackage{xspace}
\usepackage[utf8]{inputenc}
\usepackage[T1]{fontenc}
\usepackage[english]{babel}
\usepackage[left=1in,right=1in,top=1in,bottom=1in,foot=0.5in]{geometry}

\newtheorem{theorem}{Theorem}
\newtheorem{lemma}[theorem]{Lemma}
\newtheorem{claim}[theorem]{Claim}
\newtheorem{proposition}[theorem]{Proposition}

\theoremstyle{definition}
\newtheorem{remark}[theorem]{Remark}
\newtheorem{example}[theorem]{Example}

\DeclareMathOperator{\vcdim}{VC-dim}
\DeclareMathOperator{\C}{\mathcal{C}}
\DeclareMathOperator{\B}{\mathcal{B}}
\DeclareMathOperator{\diam}{diam}
\newcommand{\covectors}{\ensuremath{\mathcal{L}}}

\DeclareMathOperator{\Ima}{Im}
\DeclareMathOperator{\cB}{cB}

\newcommand{\ZZ}{\mathbb Z}
\newcommand{\bI}{\mathbf I}
\newcommand{\bR}{\mathbf R}
\newcommand{\bS}{\mathbf S}
\newcommand{\bH}{\mathbf H}

\begin{document}
\title[Sample compression schemes for balls in graphs]{Sample compression schemes for balls in graphs}

\thanks{An extended abstract of parts of this paper was presented in~\cite{MFCS_version}.}

\author[J. Chalopin]{J\'{e}r\'{e}mie Chalopin}
\address{Aix-Marseille Universit\'{e}, Universit\'{e} de Toulon, CNRS, LIS, Marseille, France}
\email{jeremie.chalopin@lis-lab.fr}

\author[V. Chepoi]{Victor Chepoi}
\address{Aix-Marseille Universit\'{e}, Universit\'{e} de Toulon, CNRS, LIS, Marseille, France}
\email{victor.chepoi@lis-lab.fr}
\author[F. Mc~Inerney]{Fionn {Mc Inerney}}
\address{Algorithms and Complexity Group, TU Wien, Vienna, Austria}
\email{fmcinern@gmail.com}
\author[S. Ratel]{S\'{e}bastien Ratel}
\address{Aix-Marseille Universit\'{e}, Universit\'{e} de Toulon, CNRS, LIS, Marseille, France}
\email{ratel.seb@gmail.com}
\author[Y. Vax\`{e}s]{Yann Vax\`{e}s}
\address{Aix-Marseille Universit\'{e}, Universit\'{e} de Toulon, CNRS, LIS, Marseille, France}
\email{yann.vaxes@lis-lab.fr}

\date{}

\begin{abstract}
One of the open problems in machine learning is whether any set-family of VC-dimension $d$ admits a sample compression scheme of size~$O(d)$. In this paper, we study this problem for balls in graphs. 
For a ball $B=B_r(x)$ of a graph $G=(V,E)$, a realizable sample for $B$ is a signed subset
$X=(X^+,X^-)$ of $V$ such that $B$ contains $X^+$ and is disjoint from $X^-$. A proper sample compression scheme of size $k$ consists of a compressor and a reconstructor. The compressor maps any realizable sample $X$  to a subsample $X'$ of size at most $k$.
The reconstructor maps each such subsample $X'$ to a ball $B'$ of $G$ such that $B'$ includes $X^+$ and is disjoint from $X^-$. 

For balls of arbitrary radius $r$, we design proper labeled sample compression schemes of size $2$ for trees, of size $3$ for cycles,
of size $4$ for interval graphs, of size $6$ for trees of cycles, and of size $22$
for cube-free median graphs.
For balls of a given radius, we design proper labeled sample
compression schemes of size $2$ for trees and of size $4$ for interval graphs.
We also design approximate sample compression schemes of size 2 for balls of $\delta$-hyperbolic graphs.
\end{abstract}

\maketitle

{}


\section{Introduction}

Sample compression schemes were introduced by Littlestone and Warmuth~\cite{LiWa}, 
and have been vastly studied in the literature due to their importance in computational machine learning.
Roughly, ~
a sample compression scheme consists of a compressor $\alpha$ and a reconstructor $\beta$, and the aim is to compress data as much as possible, such that data coherent with the original data can be reconstructed from the compressed data. For balls in graphs,  sample compression schemes of size $k$ can be defined as follows. Given a ball $B=B_r(x)$ of a graph $G=(V,E)$, a realizable sample for $B$ is a signed subset
$X=(X^+,X^-)$ of $V$ such that 
$X^+$ is included in $B$, and 
$X^-$ is disjoint from $B$. Given a realizable sample $X$, $X$ is compressed  to a subsample $\alpha(X)\subseteq X$
of size at most $k$, where the size of a (sub)sample is the number of signed elements it contains. The reconstructor $\beta$  takes $\alpha(X)$ as an input and returns $\beta(\alpha(X))$, a subset $B'$ of vertices of $G$ that is consistent with $X$, {\it i.e.},
$X^+$ is included in $B'$, and
$X^-$ is disjoint from $B'$. If $B'$ is always a ball of $G$, then the compression scheme is proper. If $X^+=B$ and $X^-=V\setminus B$, then
$\beta(\alpha(X))$ must coincide with $B$. Note that a proper sample compression scheme of size $k$ for the family of all balls of $G$ yields a sample compression scheme of size $k$ for any subfamily of balls
({\it e.g.}, for balls of a fixed radius $r$), but this scheme is no longer proper.
Sample compression schemes are labeled if $\beta$ knows the labels of the elements of $\alpha(X)$, and are unlabeled otherwise (abbreviated LSCS and USCS, resp.).
The Vapnik-Chervonenkis dimension (VC-dimension) of a set system was
introduced by Vapnik and Chervonenkis \cite{VaCh} as a complexity measure of set systems. VC-dimension is central in PAC-learning, and  
is important in combinatorics and discrete geometry. Floyd and Warmuth~\cite{FlWa} asked whether
any set-family of VC-dimension $d$ has a sample compression scheme of size~$O(d)$.
This remains one of the oldest open problems in machine learning.

In this paper, we consider the family of balls in graphs, which is as general as the sample compression conjecture. Indeed, the sample compression conjecture for set families in general is equivalent to the same conjecture restricted to the family of balls of radius 1 on split graphs in which samples only contain vertices in the clique, and the centers of the unit balls are in the stable set. Balls in graphs also constitute an important topic in graph theory, and moreover, their VC-dimension has often been considered in the literature (see, {\it e.g.},~\cite{BeDaFo,BouTh,ChEsVa,DuHaVi,PiSi}). Families of balls in Euclidean spaces have bounded VC-dimension \cite{Du1} and they constitute an important subclass of Dudley concept classes, which arise as sets of positivity for linearly parameterized functions and have been investigated in PAC-learning \cite{BDLi,Fl}. Finally,  balls in general metric spaces occur in the definition of metric entropy and metric capacity \cite{KoTi}. For relationships of these two notions with VC-dimension and PAC-learning, see~\cite{Ku}.

\subsection*{VC-dimension of balls}
The VC-dimension of the balls of radius $r$ of a graph not containing $K_{n+1}$ as a minor is at most $n$~\cite{ChEsVa}.
This result was extended to arbitrary balls in~\cite{BouTh}. Hence, the VC-dimension of balls of planar graphs is at most 4 (2 for trees and 3 for trees of cycles), and the VC-dimension of balls of a chordal graph $G$ is at most its clique number $\omega (G)$. The VC-dimension of balls of interval graphs was shown to be at most~2 in~\cite{DuHaVi}. The VC-dimension of balls of cube-free median graphs is unknown, but we prove that it is at least 4 in Example~\ref{ex:cube-free_med_VCdim}.
Finally, note that the VC-dimension of $d$-dimensional balls in $\mathbb{R}^d$ is $d+1$~\cite{Du1}. This result was generalized in \cite{Dudley,WeDu} to all Dudley classes.

\subsection*{Our results}
In this paper, we design proper sample compression schemes of small size 
for 
the family of  balls  of a graph $G$.
We investigate this problem for different graph classes. For trees, we exhibit proper USCS of size $2$ for all balls, and proper LSCS of size $2$ for balls of equal radius. We also give proper USCS of size $3$ for all balls in cycles.
For trees of cycles, we exhibit proper LSCS 
of size $6$ 
for all balls.
Then, we design proper LSCS 
of size $22$ 
for all balls of cube-free
median graphs.
We also construct proper LSCS 
of size $4$ 
for all balls of interval graphs.
This is followed by some partial results for split graphs and planar graphs.
Mainly, for any split graph $G$, we give proper LSCS of size $\omega(G)$ for all balls of $G$, and we design proper LSCS of size $4$ for balls of radius~$1$ of planar graphs.
Finally, we define $(\rho,\mu)$-approximate
proper sample compression schemes, and design ($2\delta,3\delta)$-approximate LSCS 
of size $2$ for $\delta$-hyperbolic graphs.
%
%

\subsection*{Related work}  Floyd and Warmuth~\cite{FlWa} proved that, for any concept class
of VC-dimension $d$, any LSCS 
has size at least $\frac{d}{5}$, and that,
for some maximum classes of VC-dimension $d$, they have size at least $d$.
Neylon~\cite{Neylon} proved that the concept class of positive halfspaces in $\mathbb{R}^2$ (which has VC-dimension~$2$) does not admit proper USCS of size~$2$.
Later, P\'{a}lv\"{o}lgyi and Tardos~\cite{PaTa} proved that some other concept classes of VC-dimension $2$ do not admit USCS of size at most $2$.
On the positive side, it was shown by Moran and Yehudayoff~\cite{MoYe} that LSCS of size $O(2^d)$ exist
(their schemes are not proper).  For particular concept classes, better results are known.
Floyd and Warmuth~\cite{FlWa} designed LSCS of size $d$ for regions in arrangements of central hyperplanes in ${\mathbb R}^d$. Ben-David and Litman  \cite{BDLi} obtained USCS
of size $d$ for regions in arrangements of affine hyperplanes in ${\mathbb R}^d$.
In particular, Ben-David and Litman~\cite{BDLi}
gave USCS of size $d$ for Dudley classes (see~\cite{BDLi,Dudley}). This yields USCS of size $d$ for balls in $\mathbb{R}^{d-1}$. 
Helmbold, Sloan, and Warmuth \cite{HeSlWa} (implicitly) constructed USCS
of size $d$ for intersection-closed concept classes.
Moran and Warmuth~\cite{MoWa} designed proper LSCS
of size $d$ for ample classes. Chalopin et al.~\cite{ChChMoWa} designed  USCS
of size $d$ for maximum families. They also combinatorially characterized USCS
for ample classes via the existence of
\emph{unique sink orientations} of their graphs. However, the existence of such orientations is open.
Chepoi, Knauer, and Philibert \cite{ChKnPh} extended the result of~\cite{MoWa}, and designed proper LSCS of size $d$ for concept classes defined by Complexes of Oriented Matroids (COMs). 
COMs were introduced in~\cite{BaChKn} as a natural common generalization of ample classes and Oriented Matroids \cite{BjLVStWhZi}.

\section{Definitions}\label{sec:definitions}
\subsection*{Concept classes and sample compression schemes}
Let $V$ be a non-empty finite set. Let $\C\subseteq 2^V$  be a family of subsets (also called a \emph{concept class}) of $V$. The {\it VC-dimension}
$\vcdim(\C)$ of $\C$ is the size of a largest set $Y\subseteq V$ \emph{shattered} by $\C$, {\it i.e.},
such that $\{ C\cap Y: C\in
\C\}=2^{Y}$. In machine learning,  a \emph{(labeled) sample}
is a set $X=\{(x_1,y_1),\ldots,(x_m,y_m)\}$, where $x_i\in V$ and $y_i\in\{-1,+1\}$. 
To $X$ 
is associated  the unlabeled sample
$\underline{X}=\{x_1,\ldots,x_m\}$.
A sample $X$ is \emph{realizable by a concept} $C$  if $y_i=+1$ if $x_i\in C$, and $y_i=-1$ if $x_i\notin C$.
A sample $X$ is \emph{realizable by a concept class} $\C$ if $X$ is realizable  by some $C\in \C$.

We adopt the language of sign maps and sign vectors from 
\cite{BjLVStWhZi}.
Let $\covectors$ be a  {\it set of sign vectors}, {\it i.e.},
maps from $V$ to $\{\pm 1,0\} := \{-1,0,+1\}$. The elements of $\covectors$ are
also called \emph{covectors}. For $X \in \covectors$, let
$X^+ := \{v\in V: X_v=+1\}$ and $X^-:= \{v\in V:
X_v=-1\}$. $\underline{X} = X^-\cup X^+$ 
is called the \emph{support} of $X$, and  its complement  $X^0 := V\setminus
\underline{X} = \{v\in V: X_v=0\}$ the \emph{zero set} of $X$.
Since $X^0=V\setminus (X^-\cup X^+)$, we will view any sample $X$ as $X^-\cup X^+$. Let $\preceq$ be the product ordering
on $\{ \pm 1,0\}^V$ relative to the ordering of signs with $0 \preceq -1$  and $0 \preceq +1$. 
Any concept class $\C\subseteq 2^V$ can be viewed as a set of sign vectors of $\{ \pm 1\}^V$:
for any $C\in \C$ we consider the sign vector $X(C)$, where $X_v(C)=+1$ if $v\in C$ and $X_v(C)=-1$ if  $v\notin C$.
For simplicity, we will consider $\C$ as a family of sets
and as a set of  $\{ \pm 1\}$-vectors. 
We now define sample compression schemes.  This way of presenting them seems novel.
From the definition, it follows that  a sample $X$ is just a  $\{ \pm 1,0\}$-sign vector. Given a concept class $\C\subseteq 2^V$ and $C\in \C$,
the set of samples  realizable by $C$ consists of  all covectors
 $X\in \{ \pm 1,0\}^V$ such that $X\preceq C$.  We denote by  $\downarrow\C$ the set of all samples realizable by $\C$.


A \emph{proper labeled sample compression scheme} (\emph{proper LSCS}) of size $k$ for a concept class
$\C \subseteq \{ \pm 1\}^V$ is defined by a \emph{compressor} 
$\alpha: \{ \pm 1,0\}^V \to \{ \pm 1,0\}^V$ and
a \emph{reconstructor} 
$\beta: \{ \pm 1,0\}^V  \to \C$ such that, for any
realizable sample $X\in \downarrow\C$, 
$\alpha(X)\preceq X\preceq \beta(\alpha(X))$ and $|\underline{\alpha}(X)|\le k$,
where $\preceq$ is the order between sign vectors defined above, and $\underline{\alpha}(X)$ is the support of the subsample of the sign vector $X$.
Hence, $\alpha(X)$ is a signed vector with a support of size at most $k$ such that $\alpha(X)\preceq X$, and
$\beta(\alpha(X))$  is a concept $C$ of $\C$ viewed as a sign vector. It suffices to define the map $\alpha$ only on 
$\downarrow\C$, and the map $\beta$ only on $\Ima(\alpha) := \alpha(\downarrow\C)$. The condition $X\preceq \beta(\alpha(X))$
is equivalent to the condition $\beta(\alpha(X))|\underline{X}= X$, which means that the restriction of the
concept $\beta(\alpha(X))$ to the support of $X$ coincides with the sign vector $X$.
\emph{Proper unlabeled sample compression schemes} (\emph{proper USCS}) are defined analogously, only that $\alpha(X)$ is not
a signed vector, but a subset of size at most $k$ of the support of $X$.
Formally, a proper unlabeled sample compression scheme of size $k$
for a concept class  $\C \subseteq 2^{V}$ is defined by the functions
$\alpha: \{ \pm 1,0\}^V \to 2^V$ and $\beta: 2^V \to \C$ such that, for any
realizable sample $X\in \downarrow\C$,
$\alpha(X)\subseteq \underline{X},$ $X\preceq \beta(\alpha(X))$, and $|\alpha(X)|\le k.$
For graphs, any preprocessing on the input graph $G$, such as a labeling or an embedding of $G$, is permitted and known to both the compressor and the reconstructor. As in, {\it e.g.},~\cite{LiWa,MoYe}, information, like representing the support as a vector with coordinates, is also permitted, and when we use such information, we refer to $\alpha$ and $\beta$ as vectors rather than maps.
Lastly, in our schemes, the reconstructor returns the empty set when $X^+=\varnothing$, and thus, one may consider that our schemes are not proper. We note that in all of the LSCS for the family of balls of arbitrary radius we exhibit in this paper, we could simply choose an ordering on the vertices of the graph $G=(V,E)$, and put into $\alpha(X)$ a single vertex $z\in X^-$ such that its successor $z'$ in the ordering does not belong to $X^-.$ Then, the reconstructor returns a ball $B_0(z')$ that does not intersect $X=X^-$ by the choice of $z.$ However, to avoid additional complications for such degenerate cases, we make use of the empty set.
\subsection*{Graphs}
Every graph $G=(V,E)$ in this paper is
simple and connected. 
The \emph{distance} $d(u,v):=d_G(u,v)$ between two vertices $u$ and $v$ of a graph $G$ is the length of a $(u,v)$-shortest path. The \emph{interval}
$I(u,v)$ is the set of vertices contained in $(u,v)$-shortest paths. 
A set $S$ is \emph{gated} if, for any vertex $x\in V$, there is a vertex $x'\in S$ (the \emph{gate} of $x$, with $x'=x$ if $x\in S$) such that $x'\in I(x,y)$ for any $y\in S$.
A \emph{median} of a triplet $u,v,w$ is any vertex in 
$I(u,v)\cap I(v,w)\cap I(w,u)$. A graph $G$ is \emph{median} \cite{BaCh_survey} if any triplet of vertices
$u,v,w$ has a unique median.   For any vertex $x\in V$ and any integer $r\geq 0$,
the \emph{ball of radius $r$ centered at $x$} is the
set $B_r(x) := \{ v\in V: d(v,x)\le r\}$. The unit ball $B_1(x)$ is usually denoted by $N[x]$ and called the \emph{closed neighborhood of $x$}.
The \emph{sphere} of radius $r$ centered at $x$ is the set $S_r(x)=\{ z\in V: d(z,x)=r\}$. Let also $\cB_r(u)=V\setminus B_r(u)$.
Two balls $B_{r_1}(x)$ and $B_{r_2}(y)$ are \emph{distinct} if $B_{r_1}(x)$ and $B_{r_2}(y)$ are distinct as sets.
We denote by $\B(G)$ the set of all distinct balls of $G$, and by $\B_r(G)$ the set of all distinct balls of radius $r$ of $G$.
%
For a subset $Y\subseteq V$, we call $\diam(Y)=\max\{ d(u,v): u,v\in Y\}$ the \emph{diameter} of $Y$, and we call any pair $u,v\in Y$ such that $d(u,v)=\diam(Y)$ a \emph{diametral pair} of $Y$.

\section{Trees}\label{sec:trees}

It is well known that balls in trees are gated, the family of balls has VC-dimension 2, and trees are median graphs.
For simplicity, we first consider metric trees $T$, {\it i.e.}, each edge of $T$ is homeomorphic to the segment $[0,1]$ (combinatorial trees are treated later). Indeed, metric trees will serve as a warm-up to familiarize the reader with the problem, and they cover 
the case  when the diameter of $X^+$ in a combinatorial tree is even. 

\subsection{Proper USCS for
\texorpdfstring{$\B(T)$}{ℬ(T)} for metric trees}


%
%

First, let $X$ be a realizable sample for the family $\B(T)$. 
Define $\alpha(X)$ to be the sign map
such that $\alpha(X)$ is any diametral pair 
of the set $X^+$ if $|X^+|\geq 2$, and $\alpha(X)=X^+$ otherwise. For each pair of vertices $\{u,v\}$ of $T$, let $\beta(\{ u,v\})$ be any ball $B$ of $T$ having $u,v$ as a diametral pair.
The center of $B$ is in the middle of the $(u,v)$-path, and the radius of $B$ is $\frac{1}{2}d(u,v)$. Also, let $\beta(\{u\})$ be $B_0(u)$ and $\beta(\varnothing)$ be the empty ball.

\begin{proposition} \label{USCS-trees}
For any tree $T=(V,E)$, the  pair  $(\alpha,\beta)$ of maps defines a proper unlabeled sample compression scheme
of size 2 for $\B(T)$.
\end{proposition}

\begin{proof} Let $X$ be a realizable sample for a ball $B_{r'}(y)$. If $|X^+|<2$, then $\alpha(X)=X^+$ and $\beta(\alpha(X))=X^+$, so assume that $|X^+|\geq 2$. Let $\{ u,v\}$ be a diametral pair of $X^+$, and
let $B_{r}(x)$ be the ball of $T$ returned as $\beta(\{u,v\})$. Thus, $r=\frac{1}{2}d(u,v)$. We assert that $B_{r}(x)$ is consistent with $X$. 
Otherwise, either there exists $w'\in B_{r}(x)\cap X^-$ or $w''\in X^+\setminus B_{r}(x)$. First, let there exist $w'\in B_{r}(x)\cap X^-$.
Suppose, without loss of generality, that the median  $t$ of the triplet $u,v,w'$ belongs to $I(x,v)$.
Since $w'\in B_r(x)$, $d(x,w')\le r=d(x,v)$, and therefore, $d(t,w')\le d(t,v)\le r$. Note also that $d(t,w')\le d(t,u)$ because $d(t,w')\le r$ and $d(t,u)\ge r$.
Let $z$ be the median of the triplet $u,v,y$. If $z\in I(u,t)$, then
since $d(t,w')\le d(t,v)$ and $v\in B_{r'}(y)$, we conclude that $w'\in B_{r'}(y)$, contrary to the assumption that  $w'\in X^-$. Thus, $z\in I(v,t)$.
Since $t$ is the median of $u,w',y$, and $d(t,w')\le d(t,u)$, we conclude that $d(y,w')\le d(y,u)\le r'$, contrary to the assumption that $w'\in X^-$.
Now, let there exist $w''\in X^+\setminus B_{r}(x)$. Again, let $t$ be the median of $u,v,w''$, and suppose that $t\in I(x,v)$. Since $w''\notin B_{r}(x)$,
we conclude that $d(t,w'')>d(t,v)$.  This implies that $d(u,w'')>d(u,v)$. Since $w''\in X^+$, this contradicts that $\{ u,v\}$ is a diametral pair of $X^+$.
\end{proof}

\subsection{Proper LSCS for
\texorpdfstring{$\B(T)$}{ℬ(T)} for combinatorial trees}

We present a proper labeled sample compression scheme for the family of
balls of a fixed (combinatorial) tree.

Consider a realizable sample $X$ for the family $\B(T)$ such that
$|X^+| \geq 2$. Consider a diametral pair $\{u^+,v^+\}$ of $X^+$, and let
$r = \left\lceil\frac{d(u^+,v^+)}{2} \right\rceil$.  Let $x$ be the
unique vertex in $I(u^+,v^+)$ such that
$d(u^+,x) = \left\lfloor\frac{d(u^+,v^+)}{2}\right\rfloor$ and
$d(x,v^+) = r = \left\lceil\frac{d(u^+,v^+)}{2}\right\rceil$, and $y$ be the
unique vertex in $I(u^+,v^+)$ such that
$d(u^+,y) = r = \left\lceil\frac{d(u^+,v^+)}{2}\right\rceil$ and
$d(y,v^+) = \left\lfloor\frac{d(u^+,v^+)}{2}\right\rfloor$. Note that if
$d(u^+,v^+) = 2r$, then $x = y$.

\begin{claim}\label{claimBrxBry}
  The balls $B_r(x)$ and $B_r(y)$ satisfy the following properties:
  \begin{enumerate}[(1)]
    \item $X^+\subseteq B_r(x) \cap B_r(y)$;
    \item $B_r(x) \cap X^- = \varnothing$ or
      $B_r(y) \cap X^- = \varnothing$.
  \end{enumerate}
  Consequently, one of the balls (possibly both) $B_r(x)$ or
  $B_r(y)$ realizes $X$.
\end{claim}

\begin{proof}
  The proof is similar to the proof of Proposition~\ref{USCS-trees}.
  Consider $w \in X^+$, and let $w'$ be the median of $u^+,v^+,w$. Without loss of generality, assume that
  $w' \in I(u^+,x)$. Then, $d(v^+,w) = d(v^+,x) +d(x,w) = d(v^+,y)
  +d(y,w)$. Since $u^+,v^+$ is a diametral pair of $X^+$, we have
  $d(v^+,w) \leq d(u^+,v^+) = d(u^+,x)+d(x,v^+) = d(u^+,y)+d(y,v^+)$. Consequently,
  $d(x,w) \leq d(x,u^+) \leq r$ and $d(y,w) \leq d(y,u^+) \leq
  r$. Therefore, $X^+ \subseteq B_r(x)$ and $X^+ \subseteq B_r(y)$.

  Consider a ball of center $z$ and of radius $r'$ that realizes
  $X$. Without loss of generality, assume that the median of $z,u^+,v^+$
  belongs to $I(u^+,x)$. Observe that $d(z,v^+) = d(z,x) + d(x,v^+)$, and
  thus, $r' \geq d(z,v^+) = d(z,x) + r$.  Consequently, for any
  $w \in B_r(x)$,
  $d(w,z) \leq d(w,x) + d(x,z) \leq r + d(x,z) \leq r'$, and thus,
  $B_r(x) \subseteq B_{r'}(z)$. Consequently,
  $B_r(x) \cap X^- = \varnothing$.
\end{proof}

If $d(u^+,v^+) = 2r$ is even, then $x=y$ and $B_r(x)$ realizes $X$ by
Claim~\ref{claimBrxBry}. If $d(u^+,v^+) = 2r-1$, by
Claim~\ref{claimBrxBry}, $B_r(x)$ or $B_r(y)$ realizes $X$. Observe
that $B_r(x)$ and $B_r(y)$ are the only balls of radius $r$ in $\B(T)$
containing $u^+$ and $v^+$. If $B_r(x)$ is the only ball of radius $r$
realizing $X$, then there exists $w \in X^-$ such that
$w \in B_r(y) \setminus B_r(x)$. Since $x$ and $y$ are adjacent,
necessarily, $d(w,x) = d(w,y) + 1 = r+1$. Observe that this implies
that $y$ is the median of $u^+,v^+,w$, that
$d(y,w) = d(y,u^+) = d(y,v^+)+1 = r$, and that $x$ is the unique vertex of
$I(u^+,w)$ such that $d(u^+,x) = r-1$.

\begin{proposition}
  For any tree $T$, there exists a proper labeled sample compression scheme
  of size $2$ for $\B(T)$.
\end{proposition}

\begin{proof}
  We first define the compressor $\alpha$.  Consider a realizable
  sample $X$ for the family $\B(T)$. If $|X^+| \leq 1$, then let
  $\alpha^+(X) := X^+$ and $\alpha^-(X) := \varnothing$.  Suppose now that
  $|X^+| \geq 2$, and consider the vertices $u^+, v^+, x, y$ defined as
  above, and let $r = \left\lceil \frac{d(u^+,v^+)}{2}
  \right\rceil$. If $B_r(x)$ and $B_r(y)$ both realize $X$, then
  $\alpha^+(X) := \{u^+,v^+\}$ and $\alpha^-(X):= \varnothing$.
  Otherwise, by Claim~\ref{claimBrxBry}, either $B_r(x)$ or $B_r(y)$
  realizes $X$, say the first (the other case is
  symmetric). Consequently, there exists $w \in X^-$ such that
  $d(w,x) = d(w,y) +1 = r+1$. In this case, let $\alpha^+(X) := \{u^+\}$,
  and $\alpha^-(X) := \{w\}$. Note that, in this case, $d(u^+,w) = 2r$.

  We now define the reconstructor $\beta$. Consider a sample
  $Y \in \Ima(\alpha)$. If $Y^+ = Y^- = \varnothing$, then $\beta(Y)$ is
  the empty ball. If $Y^+ =\{y\}$ and $Y^- = \varnothing$, then
  $\beta(Y) = B_0(y)$. If $Y^+ = \{u,v\}$ and $Y^- = \varnothing$, then
  let $r = \left\lceil \frac{d(u,v)}{2}\right\rceil$, and let
  $\beta(Y)$ be any ball of radius $r$ containing $u$ and $v$. Suppose
  now that $Y^+ = \{u\}$ and $Y^- = \{w\}$.  Since
  $Y \in \Ima(\alpha)$, we can assume that $d(u,w) = 2r$. Let $x$ be
  the unique vertex in $I(u,w)$ such that $d(u,x) = r-1$, and let
  $\beta(Y) = B_r(x)$.

  We claim that $(\alpha,\beta)$ is a proper labeled sample compression
  scheme for $\B(T)$. Consider a sample $X$. If $|X^+| \leq 1$, then
  $\alpha(X) = X^+$ and $\beta(\alpha(X^+)) = X^+$. Suppose now that
  $|X^+| \geq 2$. If $\alpha^+(X) = \{u^+,v^+\}$, then $u^+,v^+$ is a
  diametral pair of $X^+$, and any ball of $T$ of radius
  $r = \left\lceil \frac{d(u^+,v^+)}{2} \right\rceil$ containing $u^+$ and $v^+$ realizes
  $X$. Consequently, $\beta(\alpha(X))$ realizes $X$ in this
  case. Suppose now that $\alpha^+(X) = \{u^+\}$ and
  $\alpha^-(X) = \{w\}$. Then, in this case, the ball $B_r(x)$
  realizes $X$, where $r = \frac{d(u^+,w)}{2}$ and $x$ is the unique
  vertex in $I(u^+,w)$ that is at distance $r-1$ from $u^+$.  Since
  $\beta(\alpha(X)) = B_r(x)$, we are done.
\end{proof}

\subsection{Proper LSCS for
\texorpdfstring{$\B_r(T)$}{ℬr(T)} for combinatorial trees}

Now, let $X$ be a realizable sample for the family $\B_r(T)$. 
The sample compression scheme for $\B(T)$ cannot be applied to $\B_r(T)$ since $\beta$ may return a ball $B$ whose 
radius can be much
smaller than $r$. Keeping the center of $B$ and increasing its radius until $r$ may result in a ball which is no longer compatible with $X^-$. 
This leads us to design a different technique that encodes the center of a ball realizing the input sample. 
Let $X=X^+\cup X^-$ be a realizable sample for $\B_r(T)$. Let $\{u^+,v^+\}$ be a diametral pair of $X^+$ 
(if $X^+=\varnothing$, then $u^+$ and $v^+$ are not defined).  Similarly to the
proof of Proposition~\ref{USCS-trees}, we obtain the following result:

\begin{lemma} \label{positive-part-trees}
Any ball $B_r(x)$ of $T$ containing $\{u^+,v^+\}$  also contains $X^+$.
\end{lemma}

\begin{proof}
Let $B_r(x)$ be a ball of $T$ containing $\{u^+,v^+\}$, and suppose, by way of contradiction, that there exists $w\in X^+\setminus B_{r}(x)$. Let $t$ be the median of $u^+,v^+,w$, and suppose, without loss of generality, that $t\in I(x,v^+)$. Since $w\notin B_{r}(x)$,
we conclude that $d(t,w)>d(t,v^+)$.  This implies that $d(u^+,w)>d(u^+,v^+)$. Since $w\in X^+$, this contradicts that $\{u^+,v^+\}$ is a diametral pair of $X^+$.
\end{proof}


If $X^-$ is located far away from $X^+$, then any $r$-ball realizing $X^+$ also realizes $X$.
The next lemma provides the conditions under which this choice of the $r$-ball is no longer true:

\begin{lemma}\label{lem:r+1} 
Either any $r$-ball containing $X^+$ is disjoint from $X^-$, or there exists a
ball $B_r(x)$ realizing $X$ and a vertex $s\in X^-$ such that $d(x,s)=r+1$.
\end{lemma}

\begin{proof}
We can suppose that 
$X^-\ne \varnothing$ since otherwise any $r$-ball containing $X^+$ is trivially disjoint from $X^-$. Let $B_r(y)$ be an $r$-ball containing $X^+$ and a vertex $z\in X^-$.
Among all $r$-balls realizing $X$, pick a ball $B_r(x)$ minimizing $d(x,z)$. If $d(x,z)=r+1$, then we are done by taking $s=z$. So, let $d(x,z)>r+1$. Let $x'$ be the
neighbor of $x$ in $I(x,z)$. By the choice of $x$, the ball $B_r(x')$ does not realize $X$.
Since $d(x,z)>r+1$, $z\notin B_r(x')$. If $B_r(x')$ contains a vertex $z'\in X^-$, then $d(x,z')=r+1$,
and we are done by taking $s=z'$. Thus, there exists $w\in X^+\setminus B_r(x')$.
Since $w\in B_r(x)$, we deduce that $d(x',w)=r+1$ and $d(x,w)=r$. Consequently,  $d(w,z)=d(w,x)+d(x,x')+d(x',z)=r+1+d(x',z)>2r+1$.
This contradicts the assumption that $B_r(y)$ 
contains $X^+\cup \{ z\}$.
\end{proof}

In view of Lemma~\ref{lem:r+1}, in certain cases, the center $x$ of a ball $B_r(x)$ realizing $X$
is located on the sphere $S_{r+1}(s)$ centered at some vertex $s$ of $X^-$.
If $s$ is given, then we have to encode the position of $x$ on
$S_{r+1}(s)$ using only the vertices of $X$. Denote by  $\ell_s$ the
labeling of the sphere $S_{r+1}(s)$ obtained by performing a Depth-First Search (DFS) of $T$ with root $s$. 
This labeling assigns increasing labels (starting from 0) to the vertices of $S_{r+1}(s)$, {\it i.e.}, the first vertex of $S_{r+1}(s)$ reached by the DFS
is labeled 0, the second vertex of $S_{r+1}(s)$ reached by the DFS is labeled
1, and so on.
In other words, we consider the total order of the vertices of the tree defined by the discovery times of the DFS, and we restrict this order to the vertices of the sphere $S_{r+1}(s)$.
Let $n_s=|S_{r+1}(s)|-1$ (we also set $n_s=0$ if
$S_{r+1}(s)=\varnothing$).
Consider the labels $\{ 0,1,\ldots, n_s\}$ in clockwise circular order 
and let us call any sequence of the form $\{ i,i+1,\ldots,j-1,j\}$ or
$\{ j,j+1,\ldots,n_s,0,1,\ldots,i-1,i\}$ a \emph{circular interval}.

\begin{lemma} \label{intersection-ball-sphere} For any $v\in V(T)$,
$S_{r+1}(s)\cap B_r(v)$ and $S_{r+1}(s)\cap \cB_r(v)$ are circular intervals.
\end{lemma}

\begin{proof}
Pick $x,y\in S_{r+1}(s)\cap B_r(v)$ with $\ell_s(x)=i<\ell_s(y)=j$, and $z\in S_{r+1}(s)$ with $i<\ell_s(z)<j$.
We assert that $z\in B_r(v)$. Let $a$ be the median of  $s,x,y$. Then, $a$ is the lowest
common ancestor of $x,y$ in $T$ rooted at $s$.  According to the DFS, for $z\in S_{r+1}(s)$, we have $i<\ell(z)<j$
if and only if $a$ is an ancestor of $z$, {\it i.e.}, $a$ belongs to a path between $s$ and $z$. Since $x,y,z\in S_{r+1}(s)$,
we have $d(a,x)=d(a,z)=d(a,y)$ for any such $z$. Let $b$ be the closest to  $v$ vertex in the subtree $T'$ of $T$ spanned by the
vertices $s,x,y$. If $b\in I(s,a)$, then $a\in I(v,x)\cap I(v,y)\cap I(v,z)$, and therefore, $d(v,z)=d(v,a)+d(a,z)=d(v,a)+d(v,x)=d(v,x)\le r$, yielding $z\in B_r(v)$.
Now, suppose that $b\in I(a,x)$ (the case $b\in I(x,y)$ is similar).
By the triangle inequality and since $d(a,z)=d(a,y)$, we obtain $d(v,z)\le d(v,a)+d(a,z)=d(v,b)+d(b,a)+d(a,y)=d(v,y)\le r$, and hence, $z\in B_r(v)$ in this case too.
Finally, $S_{r+1}(s)\cap \cB_r(v)$ is an interval as the complement of $S_{r+1}(s)\cap B_r(v)$. 
\end{proof}

Let $S_{r+1}(s)\cap B_r(v)$ be a non-empty proper subset of $S_{r+1}(s)$. We
denote by $\phi_s^+(v)$ the last vertex of the circular interval $S_{r+1}(s)\cap B_r(v)$, and by $\phi_s^-(v)$ the last vertex of the circular interval $S_{r+1}(s)\cap \cB_r(v)$. That is, for the labeling $\ell_s$, $\phi_s^+(v)$ is the vertex whose label is the smallest $i \leq n_s$ such that the vertex with label $i+1 \mod (n_s+1)$ is not in the circular interval $S_{r+1}(s)\cap B_r(v)$, and $\phi_s^-(v)$ is the vertex whose label is the smallest $i \leq n_s$ such that the vertex with label $i+1 \mod (n_s+1)$ is not in the circular interval $S_{r+1}(s)\cap \cB_r(v)$.
See Fig.~\ref{fig:center_designator} for an illustration.
Let $B_r(x)$ and $s\in X^-$ be as in Lemma~\ref{lem:r+1}: $B_r(x)$ realizes the
sample $X$ and $d(x,s)=r+1$. To locate 
$x$ on the sphere $S_{r+1}(s)$, we identify $x$ by another vertex $t\in X$  such that
$\phi_s^+(t)=x$ if $t\in X^+$ or $\phi_s^-(t)=x$ if $t\in X^-$ holds. We call
$t\in X$ the {\em center designator} of the vertex $s$. 
Now, we strengthen the assertion of Lemma~\ref{lem:r+1}.

\begin{figure}[htb]
    \begin{minipage}{0.44\linewidth}
        \centering
        \includegraphics[width=0.6\linewidth]{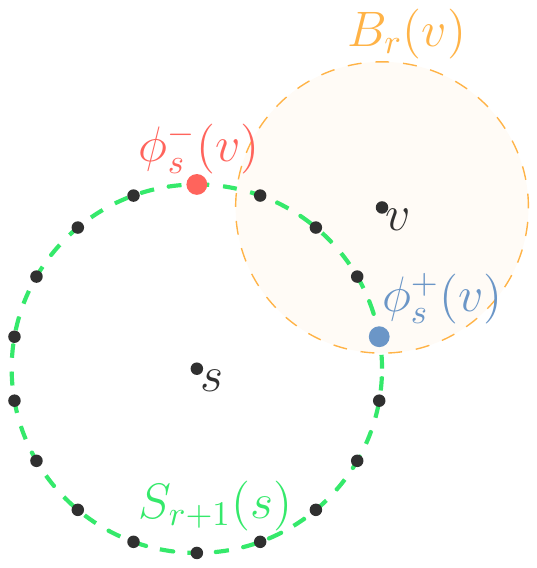}
    \end{minipage}
    \begin{minipage}{0.5\linewidth}
        \centering
        \includegraphics[width=0.65\linewidth]{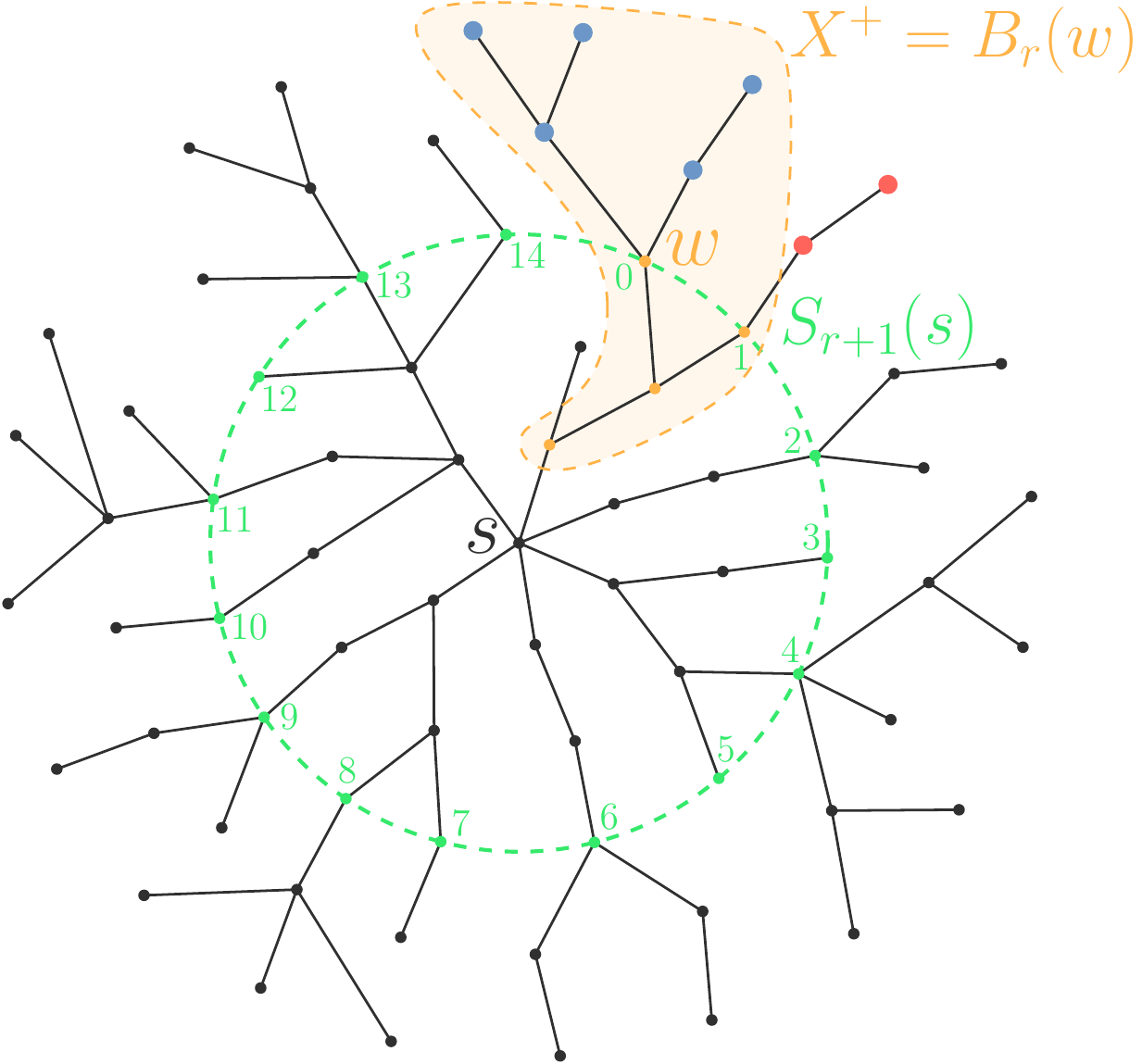}
    \end{minipage}
    \caption{
        \label{fig:phi}
        \label{fig:center_designator}
        On the left is an illustration of the vertices $\phi_s^+(v)$ and
        $\phi_s^-(v)$ associated to a vertex $v \in V(T)$ with respect to a
        vertex $s \in X^-$.
        On the right is an example of the possible center designators of a
        vertex $s \in X^-$. The vertices $t \in X^+$ such that the
        $r$-ball centered at $\phi_s^+(t)$ realizes $X$ are in blue. The
        vertices $t \in X^-$ such that the $r$-ball centered at
        $\phi_s^-(t)$ realizes $X$ are in red.
    }
\end{figure}

\begin{lemma}\label{lem:center_designator} 
One of the following conditions holds:
\begin{itemize}
\item[(1)] any $r$-ball realizing $X^+$ also realizes  $X$;
\item[(2)] there exists a vertex $s\in X^-$ such that, for all $x\in S_{r+1}(s)$, the ball $B_r(x)$ realizes $X$;
\item[(3)] there exists a vertex $s\in X^-$ having a center designator $t\in X$.
\end{itemize}
\end{lemma}

\begin{proof} By Lemma~\ref{lem:r+1}, either any $r$-ball realizing $X^+$ also realizes $X$ (and (1) holds) or there exists a ball $B_r(w)$ realizing $X$,
and a vertex $s\in X^-$ with $d(s,w)=r+1$. In the second case, if all the vertices of $S_{r+1}(s)$ are centers of $r$-balls realizing $X$, then (2) is satisfied.
Thus, we can suppose that 
$S_{r+1}(s)$ contains a vertex whose $r$-ball 
does not realize $X$. In this case, we assert that $s$ admits a center designator, {\it i.e.}, $s$ satisfies (3).
Indeed, moving clockwise along $S_{r+1}(s)$ according to the DFS-order $\ell_s$, we will find two consecutive vertices $x$ and $y$ such that $B_r(x)$ realizes the sample $X$ and $B_r(y)$ does not
realize $X$. The latter implies that either there exists $t'\in X^+\setminus B_r(y)$ or there exists $t''\in X^-\cap B_r(y)$. If $t'$ exists, since $t'\in B_r(x)$ and $t'\notin B_r(y)$, we conclude that $x\in B_r(t')$
and $y\notin B_r(t')$. Hence, $x$ is the last vertex of the circular interval $S_{r+1}(s)\cap B_r(t')$, and thus, $x=\phi_s^+(t')$, yielding that $t'$ is a center designator of $s$.
On the other hand, if $t''$ exists, since $t''\notin B_r(x)$ and $t''\in B_r(y)$, we conclude that $x\notin B_r(t'')$ and $y\in B_r(t'')$.
Thus, $x$ is the last vertex of the circular interval $S_{r+1}(s)\cap \cB_r(t'')$, and thus, $x=\phi_s^-(t'')$, yielding that $t''$ is a center designator of $s$.
\end{proof}

We now use the previous results to define the compressor $\alpha(X)$ for any realizable sample $X$. Let
$\{u^+,v^+\}$ be a diametral pair of $X^+$. We define $\alpha(X)$ according to 
Lemma~\ref{lem:center_designator}:
\begin{itemize}
    \item[(1)] if any $r$-ball realizing $X^+$ also realizes  $X$, then set
    $\alpha^-(X):=\varnothing$ and $\alpha^+(X):=(u^+,v^+)$;
    \item[(2)] if there exists a vertex $s\in X^-$ such that, for all $x\in
    S_{r+1}(s)$, the ball $B_r(x)$ realizes $X$, then set $\alpha^-(X):=(s)$ and $\alpha^+(X):=\varnothing$;
    \item[(3)] if there exists $s\in X^-$ having a center designator
    $t\in X$, then set $\alpha^-(X):=(s,t)$ and $\alpha^+(X):=\varnothing$ if $t\in X^-$,
    and set $\alpha^-(X):=(s)$ and $\alpha^+(X):=(t)$, otherwise.
\end{itemize}

In Case (3), when $|\alpha^-(X)|=2$, we suppose that the second vertex of $\alpha^-(X)$ is the center designator of the first vertex.
The reconstructor $\beta$ takes any sign vector $Y$ with a support of size at most 2 (from $\Ima(\alpha)=\alpha(\downarrow\C)$) and returns a ball $B_r(y)$ defined in the following way:
\begin{enumerate}
    \item[(1)] if $Y^-=\varnothing$, then  $\beta(Y)$ is any ball $B_r(y)$
    containing $Y^+$;
    \item[(2)] if $Y^+=\varnothing$ and $Y^-=(s)$, then $\beta(Y)$ is any
    ball $B_r(y)$ centered at a vertex $y$ of $S_{r+1}(s)$;
    \item[(3$'$)] if $Y^+=\varnothing$ and $Y^-=( s,t)$, then $\beta(Y)$ is
    the ball $B_r(y)$ centered at the vertex $y=\phi_s^-(t)$;
    \item[(3$''$)] if $Y^+=(t)$ and $Y^-=(s)$, then $\beta(Y)$ is the
    ball $B_r(y)$ centered at the vertex $y=\phi_s^+(t)$.
\end{enumerate}

From  the correspondence of the cases in Lemma~\ref{lem:center_designator} and the definitions of $\alpha$ and $\beta$, we get:

\begin{proposition} \label{LSCS-trees-extra-information}
For any tree $T=(V,E)$ and any fixed radius $r$, the pair $(\alpha,\beta)$ of vectors defines a proper labeled sample compression scheme of size 2 for $\B_r(T)$.
\end{proposition}

\begin{proof}
Let $Y=\alpha (X)$ and $B_r(y)=\beta(Y)$. Case (1) in the definition of $\beta$ corresponds to Case (1) in the definition of $\alpha$, and to Case (1) in Lemma~\ref{lem:center_designator}.
In this case, $B_r(y)$ realizes $X$ since it realizes $X^+$ by Lemma~\ref{positive-part-trees}. Case (2) in the definition of $\beta$ corresponds to Case (2) in the definition of $\alpha$, and to Case (2) in Lemma~\ref{lem:center_designator}.
In this case, $B_r(y)$ realizes $X$. Cases (3$'$) and (3$''$) in the definition of $\beta$ correspond to Case (3) in the definition of $\alpha$ and to Case (3) in Lemma~\ref{lem:center_designator}.
Thus, in  Cases (3$'$) and (3$''$), $t$ is the center designator of $s$. By Lemma~\ref{lem:center_designator} and its proof, $B_r(y)$ realizes $X$.
\end{proof}

\subsection{Proper LSCS without
  information for \texorpdfstring{$\B_r(T)$}{ℬr(T)} for combinatorial trees}

Next, we show how to design proper labeled sample compression schemes
without information for balls of fixed radius in trees. We
consider the case where the vertices of $Y$ are not distinguishable by
their order, {\it i.e.}, $\alpha(X)$ is not a
vector. 
The only case where this causes problems is when
$\alpha(X) := \{t,s\} =: Y$ and $t \in X^-$ is the center designator of
$s \in X^-$, since it is not known which of the vertices in $Y$ is the
center designator. To get past this problem, we use an arbitrary
labeling $\ell'$ of the vertices of $T$, that assigns a distinct
integer in $1,\ldots,n$ to each of the vertices of $T$, {\it i.e.},
for any two distinct vertices $u,v \in V(T)$, we have that
$\ell'(u) \neq \ell'(v)$, and we increase the size of the support to
at most $6$ in all cases. Namely, we replace item (3$''$) in the
encoding with information by the following lines (where $t$
denotes a center designator of $s$):

\begin{enumerate}
    \item[(1)] if there exists $s, t \in X^-$ such that $\ell'(s) < \ell'(t)$,
    or if $|X| = 2$, then set $\alpha^-(X) := \{s,t\}$ and $\alpha^+(X) :=
    \varnothing$;

    \item[(2)] if $|X^+| \ge 1$, then set $\alpha^-(X) := \{s,t\}$ and
    $\alpha^+(X) := \{w\}$, where $w$ is an arbitrary vertex of $X^+$;

    \item[(3)] at this point, $|X^+| = 0$ and $\ell'(t) <
    \ell'(s)$ for any center designator $t \in X^-$ of $s \in X^-$.
    If there exist $s, t, p \in X^-$ such that $\ell'(t) < \ell'(s) <
    \ell'(p)$, or if $X^- = \{s,t,p\}$, then set $\alpha^-(X) := \{s,t,p\}$;

    \item[(4)] from now on, $\ell'(s) = \max\{ \ell'(v) : v \in X^- \}$.
    If there exists $s, t, p, q \in X^-$ such that $\ell'(t) < \ell'(p) <
    \ell'(q) < \ell'(s)$, or if $X^- = \{s,t,p,q\}$, then set $\alpha^-(X)
    := \{s,t,p,q\}$;

    \item[(5)] if there exists $s, t, p, q, w \in X^-$ such that $\ell'(w) < \ell'(q) < \ell'(t) < \ell'(p)
    < \ell'(s)$, or if $X^- = \{s,t,p,q,w\}$, then set $\alpha^-(X) :=
    \{s,t,p,q,w\}$, where $q$ and $w$ are two arbitrary vertices of $X^-$
    distinct from $s$, $t$, $p$, and each other;

    \item[(6)] if $|X^-| \ge 6$, then set $\alpha^-(X) := \{s,t,p,q,w,z\}$,
    where $p$, $q$, $w$, and $z$ are four arbitrary vertices of $X^-$
    distinct from $s$, $t$, and each other.
\end{enumerate}

The reconstructor $\beta$ takes any sign vector $Y$ from $\Ima(\alpha) :=
\alpha(\downarrow\C)$ (with a support of size at most $6$) and returns a ball
$B_r$ of radius $r$ of $T$ as described in the following lines:

\begin{enumerate}

    \item[(1)] if $Y^+ = \varnothing$ and $Y^- = \{s\}$, then $\beta(Y)$ is any
    ball $B_r(y)$ centered at a vertex $y$ of $S_{r+1}(s)$;

    \item[(2)] if $|Y^+| \ge 1$ and $Y^- = \varnothing$, then $\beta(Y)$ is any ball
    $B_r(y)$ containing $Y$;

    \item[(3)] if $Y^+ = \{t\}$ and $Y^- = \{s\}$, then $\beta(Y)$ is the ball
    $B_r(y)$ centered at the vertex $y = \phi_s^+(t)$;

    \item[(4)] if $|Y| = |Y^-| = 2$, then let $s, t \in Y$ be such that
    $\ell'(s) < \ell'(t)$.
    Let $y := \phi_s^-(t)$.
    If $B_r(y) \cap Y = \varnothing$, then $\beta(Y):=B_r(y)$.
    Otherwise, $\beta(Y)$ is any ball $B_r(y)$ avoiding $Y^-$;

    \item[(5)] if $|Y^-| = 2$ and $|Y^+| = 1$, then let $s, t \in Y$ be such
    that $\ell'(s) > \ell'(t)$. Then, $\beta(Y)$ is the ball $B_r(y)$ centered
    at the vertex $y = \phi_s^-(t)$;

    \item[(6)] if $|Y| = |Y^-| = 3$, then let $s, t, p \in Y$ be such that
    $\ell'(t) < \ell'(s) < \ell'(p)$.
    Let $y := \phi_s^-(t)$.
    If $B_r(y) \cap Y = \varnothing$, set $\beta(Y):=B_r(y)$. Otherwise, $\beta(Y)$ is any ball $B_r(y)$ avoiding
    $Y^-$;

    \item[(7)] if $|Y| = |Y^-| = 4$, then let $t \in Y$ be such that $\ell'(t)
    = \min \{ \ell'(y) : y \in Y \}$, and let $s \in Y$ be such that $\ell'(s)
    = \max \{ \ell'(y) : y \in Y \}$.
    Let $y := \phi_s^-(t)$.
    If $B_r(y) \cap Y = \varnothing$, then $\beta(Y):=B_r(y)$.
    Otherwise, $\beta(Y)$ is any ball $B_r(y)$ avoiding $Y^-$;

    \item[(8)] if $|Y| = |Y^-| = 5$, then let $s, t, p, q, w \in Y$ be such that
    $\ell'(w) < \ell'(q) < \ell'(t) < \ell'(p) < \ell'(s)$.
    Let $y := \phi_s^-(t)$.
    If $B_r(y) \cap Y = \varnothing$, then $\beta(Y):=B_r(y)$.
    Otherwise, $\beta(Y)$ is any ball $B_r(y)$ avoiding $Y^-$;

    \item[(9)] if $|Y| = |Y^-| = 6$, then let $s \in Y$ be such that $\ell'(s)
    = \max \{ \ell'(y) : y \in Y \}$, and let $t \in Y$ be such that $\forall y
    \in Y, y \ne s, \ell'(y) < \ell'(t)$. Then, $\beta(Y)$ be the ball $B_r(y)$
    centered at the vertex $y = \phi_s^-(t)$.
\end{enumerate}

The following Proposition~\ref{LSCS-trees-no-extra-information} is now easy to
obtain since in line (4) ((6), resp.), either $t$
is the center designator of $s$, or $|X^+| = 0$ and $|X^-| = 2$ ($|X^-| = 3$, resp.). Moreover, in line (7) ((8), resp.), either $t$ is the center designator of $s$, or $|X^+| = 0$ and $|X^-| = 4$ ($|X^-| = 5$, resp.).

\begin{proposition} \label{LSCS-trees-no-extra-information}
    For any tree $T=(V,E)$ and any radius $r$, the pair $(\alpha,\beta)$ of
    maps defines a proper labeled sample compression scheme of size $6$ for $\B_r(T)$.
\end{proposition}

\section{Trees of cycles}\label{sec:cacti}

In this section, we present a proper labeled sample compression scheme of size 6 for balls of trees of cycles.
Recall that a \emph{tree of cycles} (or \emph{cactus}) is a graph in which each
\emph{block} (2-connected component) is a cycle or an edge. As mentioned in the introduction, the family of balls of a tree of cycles has VC-dimension $3$. In the literature, trees of cycles are classical examples of graphs of treewidth 2, and hence, difficult algorithmic problems can be solved on such graphs in polynomial time. However, the metric structure of trees of cycles is more involved (see, {\it e.g.},~\cite{GoKaPa}). As we have already seen in Section~\ref{sec:trees}, the main difficulty when considering proper labeled sample compression schemes for trees of cycles does not come from a single tree. For completeness, in this section, we first show that the difficulty for such schemes for trees of cycles also does not come from a single cycle, by giving a proper labeled sample compression scheme of size 3 for balls of cycles. In fact, as will be discussed in Remark~\ref{remark}, the difficulty comes from a spider which is a single cycle with paths of different lengths emanating from it. 

\subsection*{Cycles}\label{sec:cycles}

Given a cycle $C$, we define a proper labeled sample compression scheme of size $3$ for the set $\B(C)$ of balls of $C$ as follows. We first define the compressor function $\alpha$. For any realizable sample $X$ for $\B(C)$,
\begin{enumerate}[(1)]
\item if $|X^{+}| \leq 1$ or $|X|=|X^+|=2$, then set $\alpha(X) := X^{+}$;
\item otherwise, if $|X^{-}|=0$ and $|X^+|\geq 3$, then set $\alpha(X):= \{u,v,w\}$, where $u,v,w$ are any $3$ vertices in $X^+$;
\item otherwise ($|X^-|\geq 1$ and $|X^+|\geq 2$), there exist distinct vertices $u,v \in X^+$ and a path $P$ between $u$ and $v$ on $C$ such that $X \cap P = X^+$, and so, set $\alpha(X^+)=\{u,v\}$ and $\alpha(X^-)=\{w\}$, where $w\in X^-$ is such that $\min_{s\in X^-}\{d(u,s),d(v,s)\}=\min\{d(u,w),d(v,w)\}$.
\end{enumerate}

We now describe the reconstructor function $\beta$. Let $Y:=\alpha(X)$.
\begin{enumerate}[(1)]
\item if $Y = \varnothing$, then $\beta(Y)$ is the empty ball;
\item if $Y = \{u\}$, then $\beta(Y) = B_0(u)$;
\item if $Y = \{u,v\}$ or $Y^+=\{u,v,w\}$, then $\beta(Y)$ is any ball of $C$ covering $C$;
\item if $Y^+ = \{u,v\}$ and $Y^- = \{w\}$, then $\beta(Y)= B_r(y)$, where $P$ is the path between $u$ and $v$ on $C$ that does not contain $w$, $r=\left\lceil \frac{d_P(u,v)}{2} \right\rceil$ ($d_P(u,v)$ is the distance between $u$ and $v$ in $P$), and $y\in P$ is a center of $P$ such that $d(y,w)>r$.
\end{enumerate}

\begin{proposition} \label{LSCS-cycles}
For any cycle $C=(V,E)$, the  pair  $(\alpha,\beta)$ of maps defines a proper labeled sample compression scheme of size 3 for $\B(C)$.
\end{proposition}

\begin{proof}
We just need to show that, for any realizable sample $X$ for $\B(C)$, $\beta(\alpha(X))$ realizes $X$. If $X^+ = \varnothing$, then $\beta(\alpha(X)) = \varnothing =X^+$.  If $X^+ = \{u\}$, then $\beta(\alpha(X)) = \{u\} = X^+$.  If $|X^-|=0$ and $|X^+|\geq 2$, then $\beta(\alpha(X))=V$, and $V\cap X=X^+$. Finally, if $|X^-|\geq 1$ and $|X^+|\geq 2$, then there exist $u,v\in X^+$ and a path $P$ between $u$ and $v$ on $C$ such that $X \cap P = X^+$. Indeed, this follows since $X$ is a realizable sample for $\B(C)$, and any ball in $C$ is either a path or the entire cycle $C$. In this case, $\alpha(X^+)=\{u,v\}$ and $\alpha(X^-)=\{w\}$, where $w\in X^-$ is such that $\min_{s\in X^-}\{d(u,s),d(v,s)\}=\min\{d(u,w),d(v,w)\}$. Then, $\beta(\alpha(X))=B_r(y)$, where $P$ is the path between $u$ and $v$ not containing $w$, $r=\left\lceil \frac{d_P(u,v)}{2} \right\rceil$, and $y\in P$ is a center of $P$ such that $d(y,w)>r$. That is, $\{d(y,u),d(y,v)\}\subseteq \{r,r-1\}$, and so, $X^+\subseteq X\cap B_r(y)$. If $d_P(u,v)$ is even, then $B_r(y)=P$, and by the definition of $P$, $X\cap P=X^+$, and so, we are done. If $d_P(u,v)$ is odd, then at most one of the neighbors of $u$ and $v$ may be in $X^-$ since $X$ is a realizable sample for $\B(C)$. If one of the neighbors of $u$ and $v$ is in $X^-$, then it must be $w$ by definition. Since $B_r(y)$ contains $P$ and at most one of the neighbors of $u$ and $v$ outside of $P$, but not $w$, then $X\cap B_r(y)=X^+$.
\end{proof}


%
%
%

\subsection*{Trees of cycles}
Let $G$ be a tree of cycles.  For a vertex $v$ of $G$, let $C(v)=\{v\}$ if $v$ is a cut vertex,
and otherwise, let $C(v)$ be the unique block containing $v$.
Let $T(G)$ be the tree whose vertices are the cut vertices and the blocks of $G$, and where a cut
vertex $v$ is adjacent to a block $B$ of $G$ if and only if $v\in B$.  For any two vertices $u,v$ of $G$, let $C(u,v)$ denote the union of all cycles
and/or edges on the unique path of $T(G)$ between $C(u)$ and $C(v)$. Note
that $C(u,v)$ is a path of cycles, and that $C(u,v)$ is gated. Note also that if $u$ and $v$ belong to the same block of $G$, then $C(u,v)$ is this block.
Let $X$ be a realizable sample for $\B(G)$, and $\{u^+,v^+\}$ a diametral
pair of $X^+$. The next lemma shows that the center of a ball realizing $X$ can always be found in $C(u^+,v^+)$.

\begin{lemma}\label{lem-xinC} Let $B_r(x)$ be a ball realizing $X$, $x'$
be the gate of $x$ in $C(u^+,v^+)$, and $r'=r-d(x,x')$.
Then, the ball $B_{r'}(x')$ also realizes $X$. 
\end{lemma}

\begin{proof} Let $B':=B_{r'}(x')$. From the definition of $x'$ and $r'$, it immediately follows that $B'\subseteq B_r(x)=B$, implying that $X^-\cap B'=\varnothing$. To prove the inclusion $X^+\subseteq B'$,
suppose by way of contradiction  that there exists $z \in X^+\setminus B'$. First, suppose that $x' \in I(z,u^+)\cup I(z,v^+)$, say $x'\in I(z,u^+)$. Then, $d(u^+,v^+) \leq d(u^+,x')+d(x',v^+) \leq d(u^+,x')+r'<d(u^+,x')+d(x',z)=d(u^+,z)$, and thus, $\{ u^+,v^+\}$ is not a  diametral pair of $X^+$, a contradiction. Consequently, $x' \notin I(z,u^+)\cup I(z,v^+)$. In particular, $x'\ne u^+,v^+$. This shows that, in the graph $G\setminus\{x'\}$,
the three vertices $u^+, v^+$, and $z$ belong to the same connected component and
$x'$ is a cut vertex separating the vertex $x$ from  $u^+,v^+,z$. Since $z
\notin B_{r'}(x')$, we have
$d(x,z)=d(x,x')+d(x',z)>d(x,x')+r'=r$, contradicting the assumption  that $z \in X^+=X\cap B_r(x)$.
\end{proof}


In what follows, let $B_r(x)$ be a ball realizing $X$ with $x$ in
$C(u^+,v^+)$ 
(it exists by Lemma~\ref{lem-xinC}). If $x \in \{u^+, v^+\}$ or $x$ is
a cut vertex of $C(u^+,v^+)$, then we set $C=\{x\}$. Otherwise, $x$
belongs to a single cycle of $C(u^+,v^+)$ and  we let $C = C(x)$. The main
idea is to encode a region of $C(u^+,v^+)$ where the center $x$ of
$B_r(x)$ is located (this region may be $C$), the center, and the
radius of $B_r(x)$ by a few vertices of $X$.  The diametral pair
$\{ u^+,v^+\}$ is in $\alpha(X)$.  If $X$ contains a vertex $w$ whose
gate in $C$ is distinct from the gates of $u^+$ and $v^+$ (in $C$),
then $C$ is easily detected by including $w$ in $\alpha(X)$.
In this case, it remains to find the position
of $x$ in $C$ and to compute the radius $r$. This is done by using 2 or 3 vertices of $X$.
Otherwise, we show that $B_r(x)$ is determined by 4 vertices in $X$.

\subsection*{The partitioning of $X$}
For a vertex $y\in C(u^+,v^+)$, set $r_y:=\max\{ d(y,u^+), d(y,v^+)\}$ and
$r_y^* := \max \{ d(y,w) : w \in X^+ \}$. Clearly, $B_{r_y^*}(y)$ is
the smallest ball centered at $y$ containing $X^+$. For any vertex $z$ of $G$,
we denote by $z'$ its gate in $C(u^+,v^+)$.
Let $u^*$ and $v^*$ be the  gates of $u^+$ and $v^+$ in $C$.
If $C=\{x\}$, then $u^*=v^*=x$.
We partition $X$ and $X^-$ as follows.
Let $X_u$ ($X_u^-$, resp.) consist of all $w \in X$
($w \in X^-$, resp.) whose gate $w'$ in $C(u^+,v^+)$ belongs to
$C(u^+,u^*)$. The sets $X_v$ and $X^-_v$ are defined analogously.
Let $X_C$ ($X_C^-$, resp.) consist of all the vertices $w \in X$
($w \in X^-$, resp.) whose gates $w'$ in $C(u^+,v^+)$ belong to $C\setminus \{u^*,v^*\}$.
Note that $X_C=\varnothing$ if $C=\{x\}$.
Note also that some of these sets can be empty and that $X_u^-
\subseteq X_u$, $X_v^- \subseteq X_v$, and $X_C^- \subseteq X_C$.

If the gate in $C(u^+,u^*)$ of every vertex $w \in X_u$  is $u^+$ (this holds in
particular, if $u^*=u^+$), then we set $u_0 = u^+$. Otherwise, let
$u_0$ be the cut vertex of $C(u^+, u^*)$ farthest from $u^*$ such
that, for any vertex $w \in X_u$, its gate $w' \in C(u^+, u^*)$ is in
$C(u^+,u_0)$.  Note that we may have $u_0=u^*$.  Analogously, we
define the cut vertex $v_0$ with respect to $v^*$ and $X_v$.  Hence,
$u_0$ and $v_0$ are always well-defined.

First, suppose that $X_C = \varnothing$.
Let $w_1$ be a vertex of $X_u$ such that $w_1'=u_0$, and if no such vertex $w_1$ exists ({\it i.e.}, $u_0\neq u^+$),
then let $w_1\in X_u$ be such that $u_0\in C(w_1')$ and $w_1'$ is not a cut vertex of $C(u^+,v^+)$. Let $z_1$ be a vertex of $X^-_u$ closest to $x$.
Note that $w_1$ always exists as $u^+$ is in $X_u$, and that $z_1$
exists if and only if $X_u^-$ is non-empty.
Similarly, we define the vertices $w_2$ and $z_2$ with respect to $X_v$ and $X^-_v$.
See Fig.~\ref{fig:LSCS_cactii} for an illustration.
The next lemmas show how to compute $B_r(x)$ in this case.

\begin{figure}[htb]
    \centering
    \includegraphics[width=0.57\linewidth]{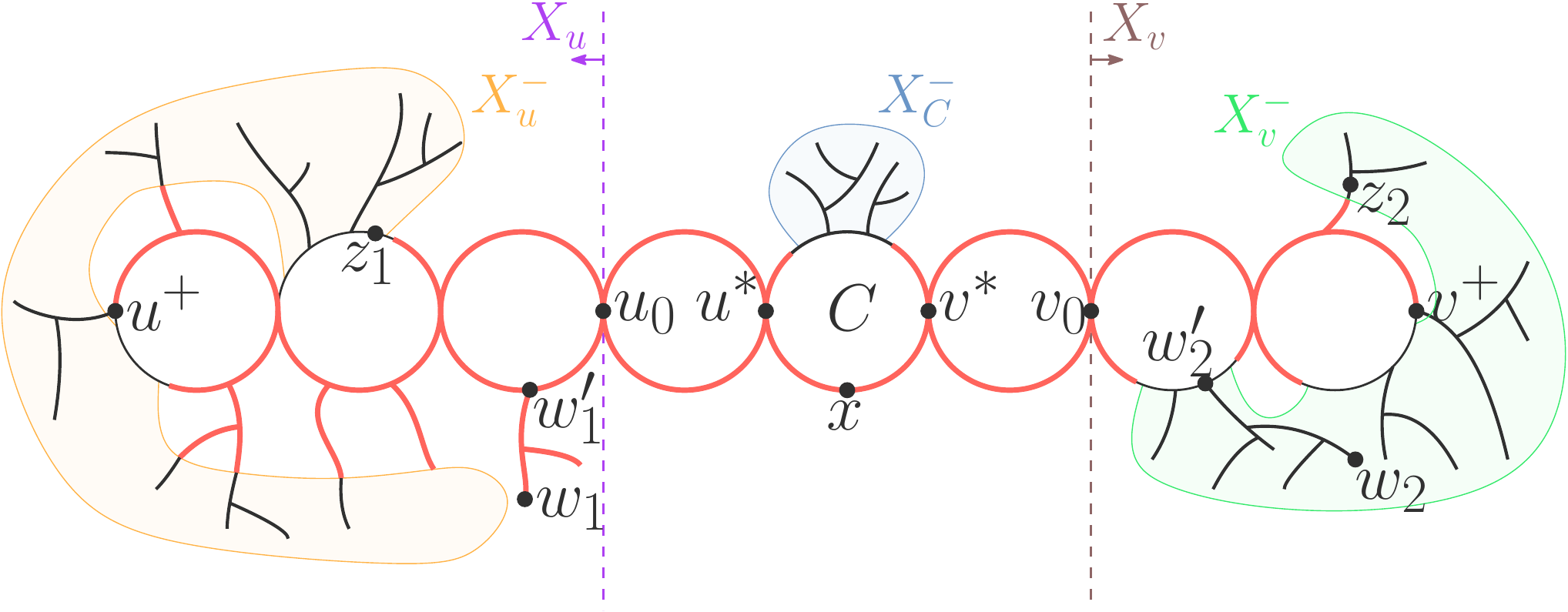}
    \caption{
        \label{fig:LSCS_cactii}
        The vertices and sets used in the proper labeled sample compression scheme for
        trees of cycles. The ball $B_r(x)$ is represented in red. 
         The cycles outside $C(u^+,v^+)$ are represented as paths.
    }
\end{figure}

\begin{lemma}
    \label{lem:ball_containing_u+v+}
    For $y \in C(u_0, v_0)$, if there exists a vertex $w \in X^+
    \setminus B_{r_y}(y)$, then $w' \in C(y)$.
    Consequently, if $X_C = \varnothing$, then, for any $y \in C(u_0,
    v_0)$, we have $X^+\subset B_{r_y}(y)$ and $r_y = r_y^*$.
\end{lemma}

\begin{proof}
Suppose that the gate $w'$ of $w$ in $C(u^+, v^+)$ does not belong to
    $C(y)$. By the definition of $u_0$ and $v_0$, it follows that $w'$
    either belongs to $C(u_0, u^+)$ or to $C(v_0, v^+)$, say $w' \in C(v_0,
    v^+)$. Then, $v_0$ separates  $u^+$ (and $y$ if $v_0\neq y$) from $v^+$ and $w$. Consequently,
    $$
    \begin{array}{lcl}
        d(y,w) = d(y,v_0) + d(v_0,w) &\text{and} &d(y,v^+) = d(y,v_0) +
        d(v_0,v^+),~\text{and} \\
        d(u^+,w) = d(u^+,v_0) + d(v_0,w) &\text{and} &d(u^+,v^+) =
        d(u^+,v_0) + d(v_0,v^+).
    \end{array}
    $$
    From the first two equalities, we obtain $d(v_0,w)>d(v_0,v^+)$. From the last two equalities, it follows that $d(u^+,w) > d(u^+,v^+)$,
    contradicting that $\{u^+, v^+\}$ is a diametral pair of $X^+$.
\end{proof}

\begin{lemma} \label{lem:ball_not_containing_z1z2}
    If $X^-\ne \varnothing$ and $X_C=\varnothing$, then $B_{r_y^*}(y)\cap
    X^-=\varnothing$ for any vertex $y\in C(u_0,v_0)$ such that
    $B_{r_y^*}(y)\cap \{ z_1,z_2\}=\varnothing$.
\end{lemma}

\begin{proof}
By Lemma~\ref{lem:ball_containing_u+v+}, we have $r_y^* = r_y$.
    Suppose by way of contradiction that there exists a vertex $z\in
B_{r_y}(y)\cap X^-$. Since $X_C=\varnothing$,
    $z$ belongs to one of the sets $X^-_u$ or $X^-_v$, say $z\in X^-_v$.
    Then, $z_2$ exists and the vertex $v_0$ separates $x$ ($y$, resp.) from $z_2$ and $z$ if $v_0\neq x$ ($v_0\neq y$, resp.).
    Consequently,
    $$
    \begin{array}{lcl}
        d(y,z_2) = d(y,v_0) + d(v_0,z_2) &\text{and} &d(y,z) = d(y,v_0) +
        d(v_0,z),~\text{and} \\
        d(x,z_2) = d(x,v_0) + d(v_0,z_2) &\text{and} &d(x,z) =
        d(x,v_0) + d(v_0,z).
    \end{array}
    $$
    Since $z_2\notin B_{r_y}(y)$ and $z\in B_{r_y}(y)$, we conclude that
    $d(v_0,z_2)>d(v_0,z)$. From the last two equalities, $d(x,z_2)>d(x,z)$,
    contrary to the choice of $z_2$ as a vertex of $X^-_v$ closest to $x$.
\end{proof}

Now, suppose that $X_C \ne \varnothing$,
in particular, $C$ is a cycle.
By the definition of $r_x^*$, $B_{r_x^*}(x)$ also realizes $X$.  Let
$w$ be a vertex of $X$ whose gate $w'$ in $C(u^+,v^+)$ is in
$C\setminus \{u^*,v^*\}$.  If, for every $y \in C$, $B_{r_y^*}(y)$
realizes $X$, then $B_{r_{w'}^*}(w')$ realizes $X$, and, in this case,
let $s \in X^+$ be such that $d(w',s) = r_{w'}^*$. Otherwise, we can
find two adjacent vertices $x$ and $y$ of $C$ such that $B_{r_x^*}(x)$
realizes $X$, but $B_{r_y^*}(y)$ does not. 
This implies that there is a vertex $z \in X^-$ with
$z \in B_{r_y^*}(y) \setminus B_{r_x^*}(x)$. In this case, let
$s, t \in X^+$ be such that $r_y^* = d(y,s)$ and $r_x^* = d(x,t)$,
with $t = s$ whenever $r_ x^* = d(x,s)$ (in
particular, this is the case  if $r_y^* = r^*_x + 1$).  Let $s'$, $t'$, and $z'$ be the
respective gates of $s$, $t$, and $z$ in $C$. If $s = t$ ($s\ne t$,
resp.), then let $P'$ be the path of $C$ between $s'$ and $z'$ ($t'$,
resp.)  containing the edge $xy$.
See Fig.~\ref{fig:r_y^*_and_r_x^*} for an illustration. 

\begin{lemma}
    \label{lem:r_y^*_and_r_x^*}
    For adjacent vertices $x, y \in C$, and the corresponding vertices $z \in X^-$ and $s
    \in X^+$, one of the following conditions holds:
    \begin{enumerate}[(1)]
        \item $r_y^* = r_x^* + 1$, $d(x,z) = d(y,z) + 1$, and $d(x,s) = d(y,s)
        - 1$;
        \item $r_y^* = r_x^* + 1$, $d(x,z) = d(y,z)$, and $d(x,s) = d(y,s) - 1$;
        \item $r_y^* = r_x^*$, $d(x,z) = d(y,z) + 1$, and $d(x,s) = d(y,s)$;
        \item $r_y^* = r_x^*$, $d(x,z) = d(y,z) + 1$, and $d(x,s) = d(y,s) - 1$.
    \end{enumerate}
\end{lemma}

\begin{proof}

Since $x$ and $y$ are adjacent, $|r_x^* - r_y^*| \le 1$. If $r_x^* = r_y^*
    + 1$, then $z$ also belongs to the ball $B_{r_x^*}(x)$, contrary to our
    choice of $x$. This show that $r_y^* \in \{r_x^*, r_x^* + 1\}$,
    establishing the first equality in each of the four cases.

    Now, notice that $x \notin I(y,z)$. Otherwise, since $r_x^* \ge r_y^* - 1$,
    we obtain that $z \in B_{r_x^*}(x)$, a contradiction.
    Since $x$ and $y$ are adjacent, we have either $d(x,z) = d(y,z) + 1$ ({\it i.e.},
    $y \in I(x,z)$) or $d(x,z) = d(y,z)$.

    If $r_y^* = r_x^* + 1$, we immediately obtain that $d(x,s) = d(y,s) - 1$
    ({\it i.e.}, $x \in I(y,s)$).
    Now, suppose that $r_y^* = r_x^*$. Since $d(y,s) = r_y^*$, we conclude that
    either $d(x,s) = d(y,s) - 1$ ({\it i.e.}, $x \in I(y,s)$) or $d(x,s) = d(y,s)$.
    This establishes the last equality in each of the four cases.

    It remains to prove that, if $r_y^* = r_x^* + 1$, then we are either in Case (1)
    or (2), and if $r_y^* = r_x^*$, then we are either in Case (3) or (4). Indeed, if
    $r_y^* = r_x^* + 1$, then $d(y,s) = d(x,s) + 1 = r_x^* + 1$, and we are in
    Case (1) or (2). On the other hand, if $r_x^* = r_y^*$, then $d(x,z) = d(y,z) + 1$ since $z
    \notin B_{r_x^*}(x)$, and we are in Case (3) or (4).
\end{proof}

Without the knowledge of $r_x^*$ and $r_y^*$, the relationships between
$d(x,z)$ and $d(y,z)$, and between $d(x,s)$ and $d(y,s)$ do not allow us to
distinguish between the cases (1) and (4). This can be done by additionally
using the vertex $t \in X^+$ defined above. Indeed, in Case (1) we have  $t = s$,
while in Case (4) we have $t \ne s$ and $d(x,t) = d(y,t) + 1$. We continue with
the following simple lemma for paths:

\begin{lemma}\label{lem:aux_lemma_P}
    Let $Q$ be a graph which is a path with end-vertices $a\ne b$, and
    let $d'$ be its distance function. Then, $Q$ contains a unique
    edge $x_0y_0$ such that $d'(x_0,b) - d'(x_0,a) \in \{1,2\}$ and
    $d'(x_0,a) < d'(y_0,a)$.
\end{lemma}

\begin{proof}
Let $k$ be a positive integer.
If the length of $Q$ is $2k$, then let
$x_0y_0$ be the edge of $Q$ such that $d'(x_0,a)=k-1$ and $d'(y_0,a)=k$.
Obviously, $d'(x_0,a) < d'(y_0,a)$.
Since $d'(x_0,b)=k+1$, we have that $d'(x_0,b) - d'(x_0,a)=2$.
If the length of $Q$ is $2k+1$, then let
$x_0y_0$ be the edge of $Q$ such that $d'(x_0,a)=k$ and $d'(y_0,a)=k+1$.
Again, $d'(x_0,a) < d'(y_0,a)$.
Since $d'(x_0,b)=k+1$, we have that $d'(x_0,b) - d'(x_0,a)=1$.
Therefore, in both cases, the edge $x_0y_0$ satisfies the distance conditions of the lemma.

Now, we prove that $x_0y_0$ is a unique such edge. If the edge $x_0y_0$ is
moved towards $a$ by $\ell\in \mathbb{N}$ hops, then $d'(x_0,b) - d'(x_0,a)$ increases by
$2\ell$, and thus, in both cases (even or odd length of $Q$), this difference will be at least $3$.
If the edge $x_0y_0$ is
moved towards $b$ by $\ell\in \mathbb{N}$ hops, then $d'(x_0,b) - d'(x_0,a)$ decreases by
$2\ell$, and thus, in both cases, this difference will be at most $0$.
Hence, in both cases, the resulting edge would not satisfy the first distance condition.
\end{proof}

\begin{figure}[htb]
    \centering
    \includegraphics[width=0.7\linewidth]{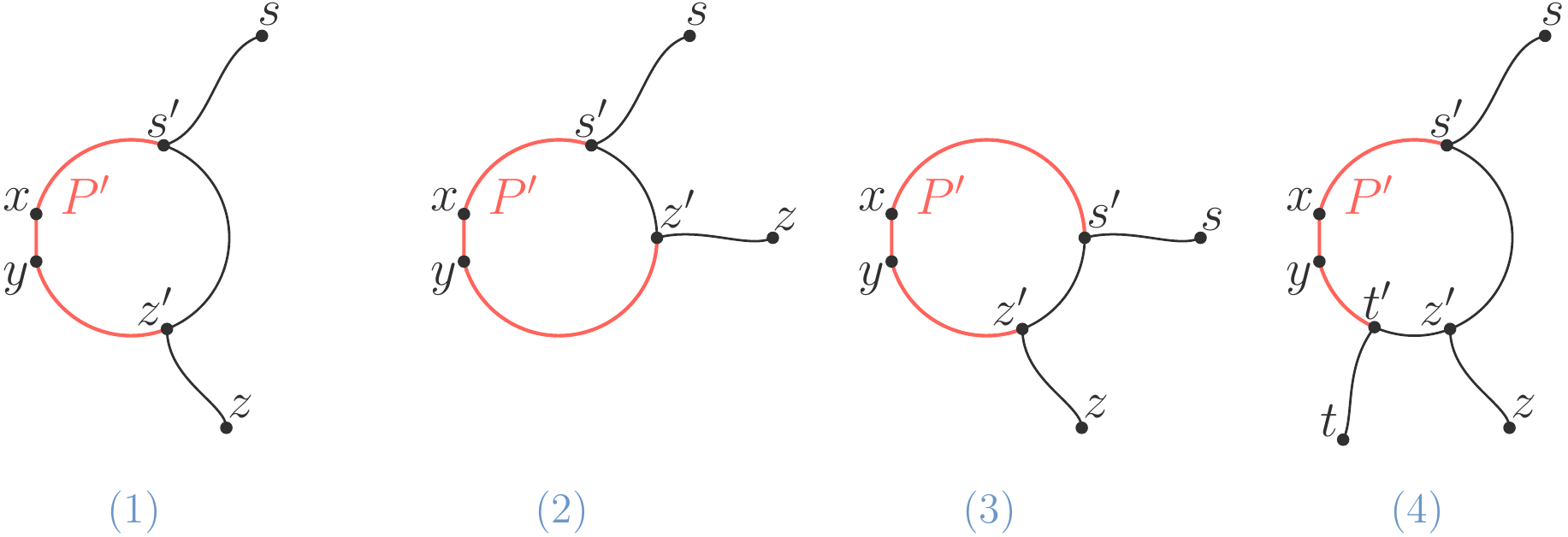}
    \caption{
        \label{fig:r_y^*_and_r_x^*}
        Definition and positioning of $s$, $t$, and $z$ in the four cases of
        Lemma~\ref{lem:r_y^*_and_r_x^*}.
    }
\end{figure}

We use Lemma~\ref{lem:aux_lemma_P} to find adjacent vertices
$x_0$ and $y_0$ of $C$ and an integer $r_x^*$ that satisfy a condition of
Lemma~\ref{lem:r_y^*_and_r_x^*}.
If $s=t$ ($s \neq t$, resp.), then, as before, let $P'$ be the path of
$C$ between $s'$ and $z'$ ($t'$, resp.) containing the edge $xy$.  Let
$P$ be the path of $G$ obtained by joining the shortest $(s',s)$- and
$(z,z')$-paths ($(s,s')$- and $(t',t)$-paths, resp.) to $P'$. Let $d'$
be the distance function on $P$.

%
%

\begin{lemma}
    \label{lem:decoding_recovering_xy}
    Let $P$ be the $(s,z)$-path or $(s,t)$-path of $G$ defined above. Let
    $x_0y_0$ be
    the unique edge of $P$  satisfying the conclusion of Lemma~\ref{lem:aux_lemma_P}. Then, $x_0 = x$ and $y_0 = y$.
    Moreover,
    \begin{enumerate}[(1)]
        \item if $P$ is an $(s,z)$-path, $d'(x_0,z) = d(x_0,z)$, and
        $d'(y_0,s) = d(y_0,s)$, then $r_x^* = d(y_0,s) - 1$;
        \item if $P$ is an $(s,z)$-path, $d'(x_0,z) = d(x_0,z) + 1$, and
        $d'(y_0,s) = d(y_0,s)$, then $r_x^* = d(y_0,s) - 1$;
        \item if $P$ is an $(s,z)$-path, $d'(x_0,z) = d(x_0,z)$, and
        $d'(y_0,s) = d(y_0,s) + 1$, then $r_x^* = d(y_0,s)$;
        \item if $P$ is an $(s,t)$-path, then $r_x^* = d(x_0,t)$.
    \end{enumerate}
\end{lemma}

\begin{proof}
By its definition, the edge $xy$ satisfies one of the four conditions of Lemma~\ref{lem:r_y^*_and_r_x^*}.
First, suppose that $xy$ satisfies Condition (1). Then, $d(x,s)=r_y^* - 1$. Since $z \in B_{r_y^*}(y) \setminus B_{r_x^*}(x)$ and
    $r_y^* = r_x^* + 1$, we conclude that $r_y^* \le d(x,z) \le r_y^* + 1$. Since $d'(x,z)=d(x,z)$ and $d'(x,s)=d(x,s)$, we conclude that
    $d'(x,z) - d'(x,s) = d(x,z) - d(x,s) \in \{1,2\}$.

    Now, suppose that $xy$ satisfies condition (2).
    Then,  $r_y^* = d(y,z)$ because $z
    \in B_{r_y^*}(y) \setminus B_{r_x^*}(x)$ and $d(x,z)=d(y,z) \le r_y^* = r_x^* +
    1$. Since $d'(x,z) = d(x,z) + 1$ and $d'(x,s)=d(x,s)$, we obtain that
    $$\begin{aligned}
        d'(x,z) - d'(x,s) &= d(x,z) + 1 - d(x,s)    \\
                          &= d(y,z) - (d(y,s) - 1) + 1\\
                          &= r_y^* - r_y^* + 2 = 2.
    \end{aligned}$$

    Now, suppose that $xy$ satisfies Condition (3).  Then,
    $d'(y,s)=d(y,s)+1=d(x,s) + 1$, $d'(x,z)=d(x,z),$ and $d'(x,s)=d(x,s)$.
    As in the previous case, we
    can show that $r_y^* = d(y,z)$. Since $r_x^* = r_y^*$ and $d(x,z) = d(y,z)
    + 1=r_y^*+1$, we conclude that
    $$\begin{aligned}
        d'(x,z) - d'(x,s) &= d(x,z) - d(x,s) \\
                          &= r_x^* + 1 - r_x^* = 1.
    \end{aligned}$$

    Since in Cases (1)--(3) of Lemma~\ref{lem:r_y^*_and_r_x^*}, we have $d'(x,z)
    - d'(x,s) \in \{1,2\}$,
    then the edge $xy$ satisfies the conclusion of Lemma~\ref{lem:aux_lemma_P}.
    Therefore, $x =
    x_0$ and $y = y_0$. We can easily check that in the first two cases,
    $r_x^* = d(y_0,s) - 1$, and in the third case, $r_x^* = d(y_0,s)$.

    Finally, if $xy$ satisfies Condition (4) of Lemma~\ref{lem:r_y^*_and_r_x^*}, then $d'(x,z)=d(x,z)$ and $d'(x,s)=d(x,s)$, and,
    as in Case (1),
    we conclude that $x = x_0$ and $y = y_0$. In this case, the path $P$ is an
    $(s,t)$-path, and therefore, $r_x^* = d(x_0,t)$ by the definition of $t$.
\end{proof}

\subsection*{The compressor $\alpha(X)$} The compressor $\alpha(X)$ is a
vector with six coordinates, which are grouped into three pairs:
$\alpha(X):=(\alpha_1(X), \alpha_2(X), \alpha_3(X))$.
The pair $\alpha_1(X) \subseteq X^+$ is a diametral pair $(u^+,v^+)$ of
$X^+$, $\alpha_2(X)$ is used to specify the region of $C(u^+,v^+)$
where the center of the target ball is located, and the pair
$\alpha_3(X)$ is used to compute the radius of this ball. We use the symbol $*$
to indicate that the respective coordinate
of $\alpha(X)$ is empty. 

We continue with the 
definitions of $\alpha_2(X)$ and $\alpha_3(X)$.
First, suppose that $X_C = \varnothing$, {\it i.e.}, $X_u \cup X_v = X$ and $X^-_u
\cup X^-_v = X^-$. 
Then, set $\alpha_2(X) := (w_1, w_2)$ and $\alpha_3(X) := (z_1, z_2)$.
Now, suppose that $X_C \ne \varnothing$. Let $w$ be a vertex of $X$ whose
gate $w'$ in $C(u^+,v^+)$ belongs to $C\setminus \{u^*,v^*\}$.
If $B_{r_x^*}(x)$ realizes $X$ for any vertex $x$ of $C$, then set
$\alpha_2(X):=(w,*)$ and $\alpha_3(X):=(s,*)$, where $s \in X^+$ is such
that $d(w',s) = r_{w'}^*$.
Otherwise, we pick an edge $xy$ of $C$ such that $B_{r_x^*}(x)$ realizes $X$
and $B_{r_y^*}(y)$ does not realize $X$.
Let  $s'$, $t'$, and $z'$ be the respective gates in $C$ of the vertices $s$,
$t$, and $z$ as defined previously.
If $s = t$, then the path $P$ is defined by the vertices $s$ and $z$, and set
$\alpha_3(X) := (s,z)$. Otherwise, the path $P$ is defined by the vertices
$s$ and $t$, and set $\alpha_3(X) := (s,t)$. Moreover, set $\alpha_2(X)
:= (*,w)$ if the edge $xy$ belongs to the path from $s'$ to $z'$
(from $s'$ to $t'$, resp.) in the clockwise traversal of $C$, and
$\alpha_2(X):=(w,*)$ otherwise. Formally, the compressor function $\alpha$ is defined in the following way:
\begin{itemize}
\item[(C1)] if $X^- = \varnothing$, set
  $\alpha_1(X) = \alpha_2(X) = \alpha_3(X) := (*,*)$;
\item[(C2)] otherwise, if $|X^+| = 0$, set
  $\alpha_1(X) = \alpha_2(X) := (*,*)$ and $\alpha_3(X) := (z,*)$,
  where $z$ is an arbitrary vertex of $X^-$;
\item[(C3)] otherwise, if $X^+ = \{u\}$, set $\alpha_1(X) := (u,*)$,
  $\alpha_2(X) := (*,*)$, and $\alpha_3(X) := (z,*)$, where $z$ is an
  arbitrary vertex of $X^-$;
\item[(C4)] otherwise, if $|X^+| \ge 2$ and $X_C = \varnothing$, set
  $\alpha_1(X) := (u^+, v^+)$, $\alpha_2(X) := (w_1,w_2)$, and
  \begin{itemize}
  \item[(C4i)] if the vertex $z_2$ does not exist, then set $\alpha_3(X)
    := (z_1, *)$;
  \item[(C4ii)] if the vertex $z_1$ does not exist, then set $\alpha_3(X)
    := (*, z_2)$;
  \item[(C4iii)] if the vertices $z_1$ and $z_2$ exist, set $\alpha_3(X)
    := (z_1, z_2)$;
  \end{itemize}
\item[(C5)] otherwise ($|X^+| \ge 2$ and $X_C \ne \varnothing$), and
  \begin{itemize}
  \item[(C5i)] if, for any vertex $y \in C$, the ball $B_{r_y^*}(y)$
    realizes $X$, then set $\alpha_1(X) := (u^+,v^+)$,
    $\alpha_2(X) := (w,*)$, and $\alpha_3(X) := (s,*)$, where
    $s \in X^+$ is such that $d(w',s) = r_{w'}^*$;
  \item[(C5ii)] otherwise, if $s=t$ and the edge $xy$
    belongs to the clockwise $(s',z')$-path of $C$, then set
    $\alpha_2(X) := (*,w)$ and $\alpha_3(X) := (s,z)$;
  \item[(C5iii)] otherwise, if $s=t$ and the edge
    $xy$ belongs to the counterclockwise $(s',z')$-path of $C$, then
    set $\alpha_2(X) := (w,*)$ and $\alpha_3(X) := (s,z)$;
  \item[(C5iv)] otherwise ($s\neq t$), if the edge $xy$
    belongs to the clockwise $(s',t')$-path of $C$, then set
    $\alpha_2(X) := (*,w)$ and $\alpha_3(X) := (s,t)$;
  \item[(C5v)] otherwise ($s\neq t$), if the edge $xy$
    belongs to the counterclockwise $(s',t')$-path of $C$, then set
    $\alpha_2(X) := (w,*)$ and $\alpha_3(X) := (s,t)$.
  \end{itemize}
\end{itemize}
\subsection*{The reconstructor  $\beta(X)$}
Let $Y$ be a vector on six coordinates grouped into three pairs $Y_1$, $Y_2$,
and $Y_3$.
If $Y_1 = (y_1,y_2)$, then, for any vertex $t$ of $G$, we denote by $t'$ its gate
in $C(y_1,y_2)$. For any vertex $y$ of $C(y_1,y_2)$, we also set $r_y :=
\max\{d(y,y_1), d(y,y_2)\}$.
The reconstructor $\beta$ takes $Y$ and returns a ball $B_r(y)$ of $G$ defined
in the following way:
\begin{enumerate}
\item[(R1)] if $Y = ((*,*),(*,*),(*,*))$, then $\beta(Y)$ is any ball
  that contains the vertex set of $G$;
\item[(R2)] if $Y = ((*,*),(*,*),(y_5,*))$, then $\beta(Y)$ is the
  empty set;
\item[(R3)] if $Y = ((y_1,*),(*,*),(y_5,*))$, then $\beta(Y)$ is the ball
  $B_0(y_1)$;
  \item[(R4)] if $Y_1 = (y_1,y_2)$ and $Y_2 = (y_3,y_4)$, then let $u_0=y_3'$ if $y_3' \in  \{y_1,y_2\}$ or $y_3'$ is a cut vertex of $C(y_1,y_2)$, and otherwise, let $u_0$
  be the cut vertex of $C(y_3')$ between $y_3'$ and $y_2$. Similarly, let
  $v_0=y_4'$ if $y_4'\in \{y_1,y_2\}$ or $y_4'$ is a cut vertex of $C(y_1,y_2)$, and otherwise, let $v_0$ be the cut vertex of $C(y_4')$ between $y_4'$ and $y_1$.
  Then, $\beta(Y)$ is any ball $B_{r_y}(y)$  centered at $y \in
  C(u_0,v_0)$ such that $B_{r_y}(y)$ contains no vertex of $Y_3$.
\item[(R5i)] if $Y = ((y_1,y_2),(y_3,*),(y_5,*))$, then $\beta(Y)$ is
  the ball $B_r(y_3')$ of radius $r = d(y_3',y_5)$;
\item[(R5ii)] if $Y = ((y_1,y_2),(*,y_4),(y_5,y_6))$ and
  $(y_5,y_6) \in X^+ \times X^-$, let $xy$ be the edge of the
  $(y_5',y_6')$-path in the clockwise traversal of the cycle $C(y_4')$
  such that $|d'(x,y_6) - d'(x,y_5)| \in \{1,2\}$ and $y$ is closer to
  $y_6$ than $x$ is.
  Let $\beta(Y)$ be the ball
  $B_r(x)$, where $r = d(y,y_5)$ if $d'(y,y_5) = d(y,y_5) + 1$, and
  $r = d(y,y_5) - 1$ otherwise;
\item[(R5iii)] if $Y = ((y_1,y_2),(y_3,*),(y_5,y_6))$ and
  $(y_5,y_6) \in X^+ \times X^-$, let $xy$ be the edge of the
  $(y_5',y_6')$-path in the counterclockwise traversal of $C(y_3')$ such
  that $|d'(x,y_6) - d'(x,y_5)| \in \{1,2\}$ and $y$ is closer to
  $y_6$ than $x$ is.
  Let $\beta(Y)$ be the ball $B_r(x)$, where
  $r = d(y,y_5)$ if $d'(y,y_5) = d(y,y_5) + 1$, and $r = d(y,y_5) - 1$
  otherwise;
\item[(R5iv)] if $Y = ((y_1,y_2),(*,y_4),(y_5,y_6))$ and
  $(y_5,y_6) \in X^+ \times X^+$, let $xy$ be the edge of the
  $(y_5',y_6')$-path in the clockwise traversal of the cycle $C(y_4')$
  such that $|d'(x,y_6) - d'(x,y_5)| \in \{1,2\}$ and $y$ is closer to
  $y_6$ than $x$ is.
  Let $\beta(Y)$ be the ball $B_r(x)$, where $r = d(x,y_6)$;
\item[(R5v)] if $Y = ((y_1,y_2),(y_3,*),(y_5,y_6))$ and
  $(y_5,y_6) \in X^+ \times X^+$, let $xy$ be the edge of the
  $(y_5',y_6')$-path in the counterclockwise traversal of the cycle
  $C(y_3')$ such that $|d'(x,y_6) - d'(x,y_5)| \in \{1,2\}$ and $y$ is
  closer to $y_6$ than $x$ is.
  Let $\beta(Y)$ be the ball $B_r(x)$, where $r = d(x,y_6)$;
\end{enumerate}


\begin{proposition} \label{LSCS-cacti-extra-information}
For any tree of cycles $G$, the pair $(\alpha,\beta)$ of vectors defines a proper labeled sample compression scheme of size 6 for $\B(G)$.
\end{proposition}

\begin{proof}
 Let $X$ be a realizable sample for $\B$.
    Let $Y=\alpha (X)$ and $B_r(x^*)=\beta(Y)$. We will prove case by case that
    the ball $B_r(x^*)$ realizes the sample $X$, {\it i.e.}, that $X^+ \subseteq
    B_r(x^*)$ and $X^-
    \cap B_r(x^*)=\varnothing$. One can easily see that the cases (R$k$) and
    their subcases in the definition of $\beta$ correspond to the cases (C$k$)
    and their subcases in the definition of $\alpha$: namely, the vector $Y$ in
    Case (R$k$) has the same specified coordinates as the vector $\alpha(X)$ in
    Case (C$k$).

    In Case (R1), we have $Y = ((*,*),(*,*),(*,*))$. Since $Y=\alpha(X)$, this
    implies that $|X^-|=0$, which corresponds to Case (C1). Consequently, the
    ball covering $G$ is compatible with $X$. In Case (R2), $Y =
    ((*,*),(*,*),(y_5,*))$. Since $|Y^+| = 0$ and $|Y^-|\ne 0$, this implies that
    $|X^+|=0$ and $|X^-| \ne 0$, which corresponds to Case (C2). In this case,
    the empty set is compatible with $X$.
    In Case (R3), $Y = ((y_1,*),(*,*),(y_5,*))$. The fact that $|Y_1|=1$ implies that $|X^+| = 1$.
    Thus, the ball $B_0(y_1)$ is compatible with $X$.

    Consider now Case (R4). In this case, $Y_1 = (y_1,y_2)$ and $Y_2 =
    (y_3,y_4)$.
    Since $Y=\alpha(X)$, this implies that $X$ satisfies the conditions of Case
    (C4), {\it i.e.},  $|X^+|\ge 2$ and $X_C=\varnothing$.
    Therefore, $Y_1 = (y_1,y_2)=(u^+,v^+)=\alpha_1(X), Y_2 =
    (y_3,y_4)=(w_1,w_2)=\alpha_2(X)$, and $Y_3$ (containing one or two
    vertices)
    coincides with $\alpha_3(X)$ (containing one or two vertices $z_1,z_2$ as
    in subcases (C4i)--(C4iii)). The ball $B_{r_y}(y)$ returned by  Case (R4)
    is centered at $y \in C(u_0,v_0)$, contains $Y_1$, and is disjoint from $Y_3$. Since the target ball $B_r(x)$ has its center on $C(u_0,v_0)$
    and is compatible with $X\supseteq Y_1\cup Y_3$, the ball $B_{r_y}(y)$ is well-defined. By Lemma~\ref{lem:ball_containing_u+v+}, $B_{r_y}(y)$ contains $X^+$.
    By Lemma~\ref{lem:ball_not_containing_z1z2}, $B_{r_y}(y)$ is disjoint from
    $X^-$. Consequently, $B_{r_y}(y)$ is compatible with $X$.

    In Case (R5), we have $Y_2 = (*,y_4)$ or $Y_2 = (y_3,*)$. This
    distinguishes the Case (R5) from the Cases (R1)--(R4). If $Y_3 = (y_5,*)$,
    then we are in Case (R5i), which corresponds to Case (C5i). In this
    case, from the definition of $y_3'$ and $y_5$, we conclude that $r_{y_3'}^*
    = d(y_3',y_5)$, and thus, the ball $\beta(Y) = B_{r_{y_3'}^*}(y_3')$ realizes
    $X$.
    Each subcase (R5ii)--(R5v) corresponds to the respective subcase
    (C5ii)--(C5v), and its analysis is based on Lemma~\ref{lem:decoding_recovering_xy}. If $Y_3 = (y_5,y_6) \in X^+ \times
    X^+$, then we are in Case (R5iv) or (R5v), and also in Case (4) of Lemma~\ref{lem:decoding_recovering_xy}.
    Therefore, $\beta(Y) = B_r(x)$ with $r = d(x,y_6)$ realizes $X$.
    If $Y_3 = (y_5,y_6) \in X^+ \times X^-$, then we are in Case (R5ii) or
    (R5iii), and in Cases (1)--(3) of Lemma~\ref{lem:decoding_recovering_xy}. If
    $d'(y,y_5) = d(y,y_5) + 1$, then we are in Case (3) of Lemma~\ref{lem:decoding_recovering_xy}. Therefore, by Lemma~\ref{lem:decoding_recovering_xy}, $r_x^* = d(y,y_5) =: r$, and the
    ball $\beta(Y) = B_r(x)$ realizes $X$. Otherwise, $d'(y,y_5) = d(y,y_5)$,
    and we are in Case (1) or (2) of Lemma~\ref{lem:decoding_recovering_xy}. In
    these cases, if we set $r = d(y,y_5) -
    1$, then the ball $B_r(x)$ realizes $X$.
\end{proof}

\begin{remark}\label{remark}
    The most technically involved case of the previous result is the case $X_C
    \ne \varnothing$. In fact, this case corresponds to proper labeled sample compression schemes in
    \emph{spiders}, {\it i.e.},
    in  graphs consisting of a single cycle $C$ and paths of different lengths
    emanating from this cycle. Due to this case,  $\alpha(X)$ in our result is not a signed map
    but a signed vector of size $6$. Thus, in this case, we need extra
    information compared to the initial definition of proper labeled sample
    compression schemes. The VC-dimension of the family of balls in a spider and
    in a tree of cycles is $3$. We wonder \emph{whether the family of balls in
    spiders admits a proper labeled sample compression scheme without any information that is of (a)
    size $3$ or (b) constant size.}
\end{remark}

\section{Cube-free median graphs}\label{sec:cube-free}

The \emph{dimension} $\dim(G)$ of a median graph $G$ is the largest dimension
of a hypercube of $G$. A \emph{cube-free median graph} is a median graph of dimension $2$, {\it i.e.},
a median graph not containing 3-cubes as isometric subgraphs. Median graphs constitute the most important class of graphs in metric graph theory.
They are also important in geometric group theory (as the 1-skeletons of CAT(0) cube complexes) and in concurrency (as the domains of event structures).
For references about median graphs, see~\cite{BaCh_survey}. For cube-free median graphs,
see \cite{BaChEp,ChHa,ChLaRa,ChMa}. We use the fact that intervals of median graphs are convex, and thus, gated~\cite{Mu80}. We describe a
proper LSCS of size 22 for balls of cube-free median graphs. The following example shows that the balls of cube-free median graphs have VC-dimension at least $4$.

\begin{example}\label{ex:cube-free_med_VCdim}
The balls of cube-free median graphs have VC-dimension at least $4$. Let $G$ be the cube-free median graph in Figure~\ref{fig:cube-free_med_VCdim}. We show that the set $Y\subset V(G)$ of $4$ vertices with labels $1,2,3,4$ can be shattered by $\B(G)$. It is trivial to see there are balls whose intersection with $Y$ are precisely the empty set, each subset of size $1$ of $Y$, and $Y$. For $\{1,2\}$, $\{1,3\}$, $\{1,4\}$, $\{2,3\}$, $\{2,4\}$, and $\{3,4\}$, respectively, take the balls $B_2(c)$, $B_3(b)$, $B_3(a)$, $B_1(g)$, $B_2(d)$, and $B_2(e)$, respectively. For $\{1,2,3\}$, $\{1,2,4\}$, $\{1,3,4\}$, and $\{2,3,4\}$, respectively, take the balls $B_3(f)$, $B_3(c)$, $B_4(b)$, and $B_2(g)$, respectively.
\end{example}

\begin{figure}[htb]
\includegraphics[scale=0.75]{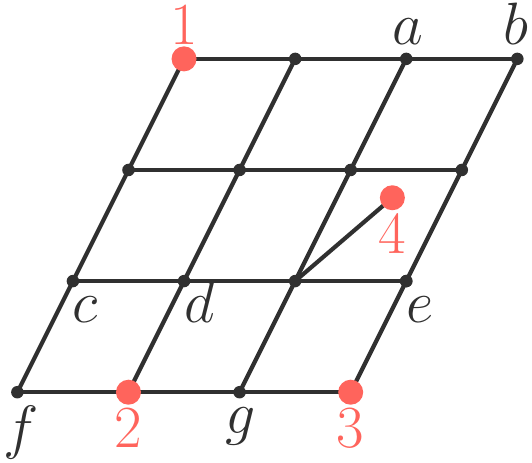}
\caption{\label{fig:cube-free_med_VCdim} Illustration of the graph $G$ in Example~\ref{ex:cube-free_med_VCdim} that is used to show that the balls of cube-free median graphs have VC-dimension at least $4$. The set of $4$ vertices in red with labels $1,2,3,4$ are shattered by $\B(G)$.}
\end{figure}

Let $G$ be a cube-free median graph. Let $X$ be a realizable sample for $\B(G)$, and $\{u^+,v^+\}$ a diametral
pair of $X^+$. The next lemma shows that the center of a ball realizing $X$
can always be found in $I(u^+,v^+)$ (this result does not hold for all median graphs):
\begin{lemma}\label{lem_xinC_cube-free} 
If $x'$ is the gate of $x$
    in the interval $I(u^+,v^+)$, and $r' = r - d(x,x')$, then $X$ is a
    realizable sample for $B_{r'}(x')$, {\it i.e.}, $X^+\subseteq B_{r'}(x')$ and
    $X^-\cap B_{r'}(x')=\varnothing$.
\end{lemma}

\begin{proof}
    Let $B':=B_{r'}(x')$. From the definitions of $x'$ and $r'$, it immediately
    follows that $B'\subseteq B_r(x)=B$, implying that $X^-\cap
    B'=\varnothing$. Notice also that $d(x',u^+) \le r'$ and $d(x',v^+) \le r'$.
    To prove the inclusion $X^+\subseteq B'$, suppose by way of contradiction
    that there exists $z \in X^+\setminus B'$.
    Let $z'$ be the gate of $z$ in $I(u^+,v^+)$.
    First, suppose that $z' \in I(x,z)$. Since $x' \in I(x,z')$, the vertices
    $z'$ and $x'$ belong to a common shortest path between $z$ and $x$.
    Consequently,
    $$d(x,z) = d(x,x') + d(x',z') + d(z',z).$$
    Since $d(x,z) \le r$, we conclude that $d(x',z) \le r - d(x,x') = r'$.
    The same conclusion holds if $x' \in I(x,z)$.
    This implies that the median of $x$, $z$, and $z'$ is a vertex $z'' \ne z'$.

    Now, suppose that the vertices $x'$ and $z'$ belong to a common shortest
    path between $u^+$ and $v^+$. Say $z' \in I(u^+,x')$ and $x' \in I(z',v^+)$.
    Since $d(x',u^+) \le r'$ and $d(x',z) > r'$, we deduce that $d(x',u^+) <
    d(x',z)$. Since $z' \in I(x',u^+) \cap I(x',z)$, we deduce that $d(z,z') >
    d(u^+,z')$. Since $z' \in I(v^+,u^+) \cap I(v^+,z)$, we conclude that $d(v^+,z) >
    d(v^+,u^+)$, contradicting the  choice of the pair $\{u^+,v^+\}$ as a diametral
    pair of $X^+$.

    Finally, suppose that $x'$ and $z'$ do not belong to a common shortest path
    between $u^+$ and $v^+$. This is equivalent to the assertion that the median
    $u'$ of the triplet $\{u^+,x',z'\}$ is different from $x'$ and $z'$, and that
    the median $v'$ of the triplet $\{v^+,x',z'\}$ is different from $x'$ and
    $z'$.
    Notice that $u'$ and $v'$ belong to a common shortest path between $u^+$ and
    $v^+$. Let $p$ be a neighbor of $z'$ in $I(z',u')$ and let $q$ be a neighbor
    of
    $z'$ in $I(z',v')$. If $p = q$, then this vertex is the median of
    $(z',u',v')$ and of $(z',u^+,v^+)$, contrary to the assumption that $z' \in
    I(u^+,v^+)$. It follows that $p \ne q$. Let $s$ be a neighbor of $z'$ in
    $I(z',z'')$. Since $z'$ is the gate of $z$ in $I(u^+,v^+)$, $s$ does not belong
    to $I(u^+,v^+)$, and is thus distinct from $p$ and $q$.
    From the definitions of the  vertices $x'$, $z'$, and $z''$, we conclude that
    the vertices $p$, $q$, and $s$ belong to the interval $I(z',x)$ and therefore belong to a 3-cube of $G$. This
    contradicts the fact that $G$ is a cube-free median graph. 
\end{proof}

%
%
By \cite{ChMa},  $I(u^+,v^+)$ of a cube-free median graph has an isometric embedding in
the square grid $\ZZ^2$. We denote by $(z_a,z_b)$ the coordinates in $\ZZ^2$ of a vertex $z\in I(u,v)$.  We consider
isometric embeddings of $I(u,v)$ in $\ZZ^2$ for which $u=(0,0)$ and $v=(v_a,v_b)$ with $v_a\ge 0$ and $v_b\ge 0$. 
We fix 
a canonical isometric embedding, which can be used
both by the compressor and the reconstructor. 
Finally, we use the same notation for the vertices and their images under the
embedding, and we denote by $\bI$ the interval $I(u^+, v^+)$ embedded in
$\ZZ^2$. As usual, for a vertex $z \in V$, we denote by $z'$ its gate in the interval $I(u^+,v^+)$.
\subsection*{The compressor $\alpha(X)$}
The compressor $\alpha(X)$ is a vector with $22$ coordinates grouped
into four parts $\alpha(X):=(\alpha_1(X),\alpha_2(X),\alpha_3(X),\alpha_4(X))$.
The part $\alpha_1(X) \subseteq X^+$ consists of a diametral pair $(u^+,
v^+)$ of $X^+$. The part $\alpha_2(X) \subseteq X$ has size 4, and is
used to specify a region $\bR \subseteq \bI=I(u^+, v^+)$ such that the gates in
$I(u^+, v^+)$ of all the vertices of $X$ are located outside or on the boundary of
$\bR$. Moreover, $\bR$ contains the center $x$ of the target ball
$B_r(x)$.
The parts $\alpha_3(X) \subseteq X^+$ and $\alpha_4(X)\subseteq X^-$ each have size 8 and are used
to locate the center 
and the radius 
of a ball $B_{r''}(y)$ realizing $X$. Now, we formally define $\alpha_i(X)$, $i = 1,...,4$. 
Let $X_1 := \{ w \in X : w'_b \ge x_b \}$, $X_2 := \{ w \in X : w'_a \ge x_a
\}$, $X_3 := \{ w \in X : w'_b \le x_b \}$, and $X_4 := \{ w \in X : w'_a \le
x_a \}$. Since $I(u^+,v^+)$ is gated, $X=\cup_{i=1}^4 X_i$.
Denote by $X'_i$, $i=1,...,4$, the gates of the vertices of $X_i$ in
$I(u^+, v^+)$. Set $\alpha_2(X):= (w_1, w_2, w_3, w_4) \in X^4$, where:
\begin{itemize}
    \item $w_1$ is a vertex of $X_1$ whose gate $w_1'$ has the smallest
    ordinate among the vertices of $X_1'$;
    \item $w_2$ is a vertex of $X_2$ whose gate $w_2'$ has the smallest
    abscissa among the vertices of $X_2'$;
    \item $w_3$ is a vertex of $X_3$ whose gate $w_3'$ has the largest
    ordinate among the vertices of $X_3'$;
    \item $w_4$ is a vertex of $X_4$ whose gate $w_4'$ has the largest
    abscissa among the vertices of $X_4'$;
\end{itemize}
%
%
%
%
%
\begin{figure}[htb]
\begin{minipage}{0.495\linewidth}
        \centering
        \includegraphics[width=0.538\linewidth]{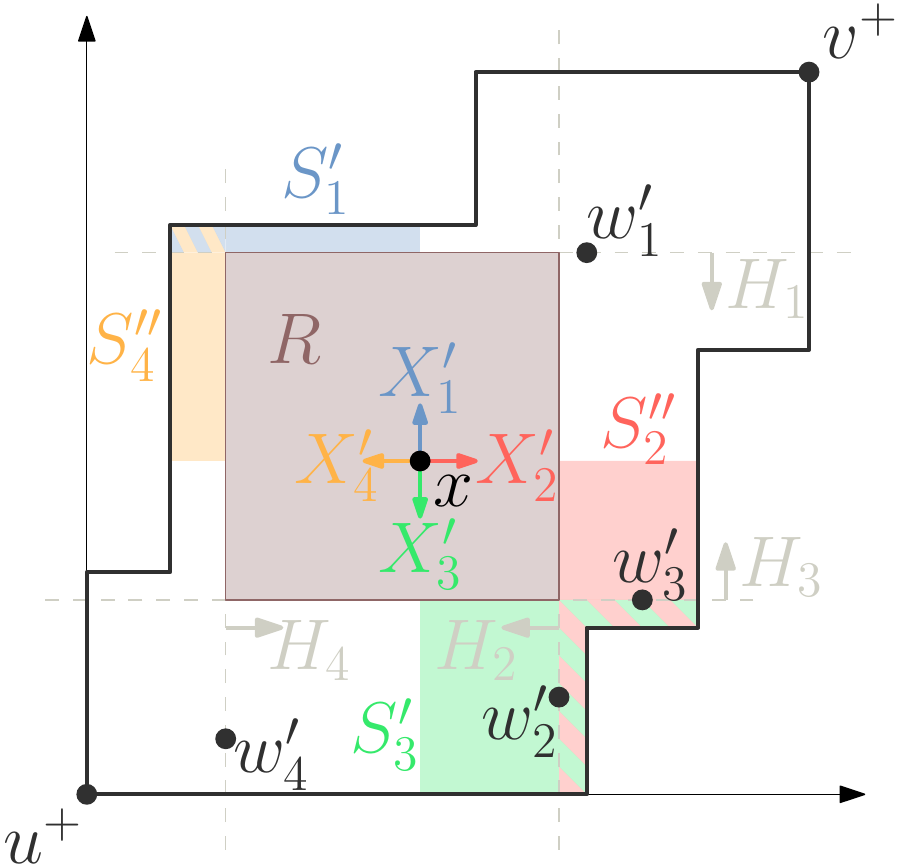}
    \end{minipage}
    \begin{minipage}{0.495\linewidth}
        \centering
        \includegraphics[width=0.66\linewidth]{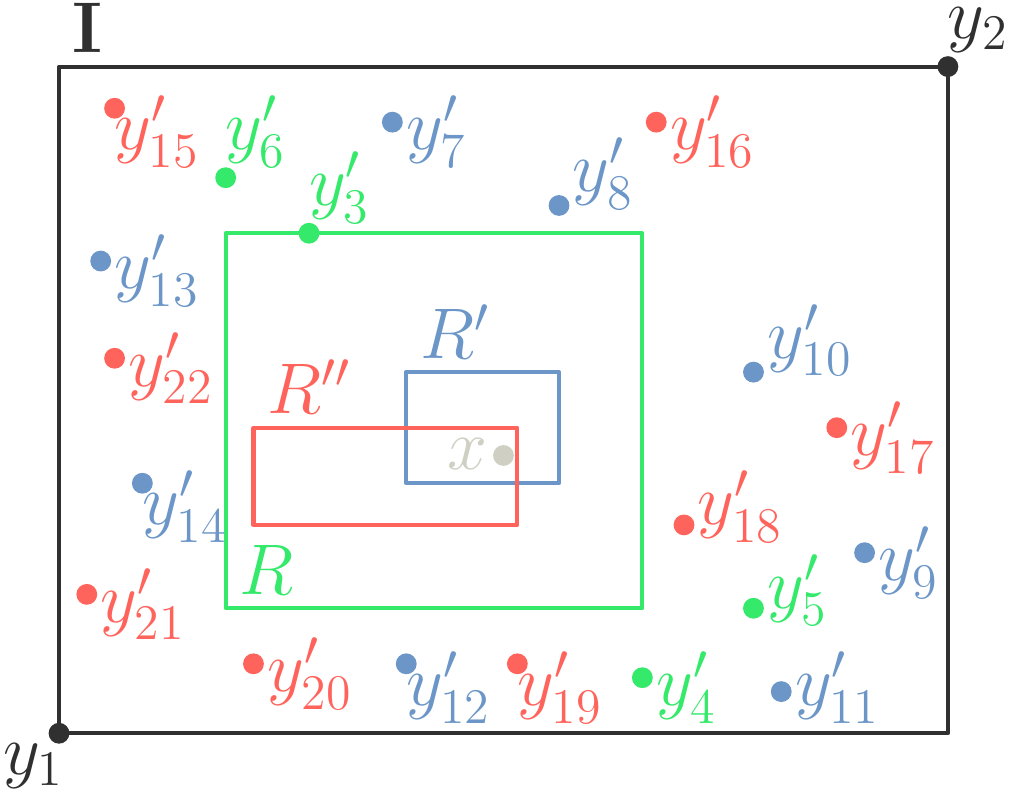}
    \end{minipage}
    \caption{
        \label{fig:LSCS_cube-free_median}
        On the left, the region $\bR$ and the halfstrips $\bS_1'(x)$,
        $\bS_2''(x)$, $\bS_3'(x)$, and $\bS_4''(x)$. 
        On the right, the regions $\bR$, $\bR'$, and $\bR''$ computed from
        $\alpha(X)$.
        Steps 1-4 of the reconstruction correspond to the
        black, green, blue, and red parts of the figure.
        The target center $x$ is given in gray. 
         }
\end{figure}
For a vertex $w=(w_a, w_b)\in \ZZ^2$, we consider the four coordinate
halfplanes $\bH_{\le w_a} := \{t : t_a \le w_a\}$, $\bH_{\ge w_a}$,
$\bH_{\le w_b}$, 
and $\bH_{\ge w_b}$. 
Let $\bR$ be the set of vertices of $\bI$ that belong to the intersection of
the halfplanes $\bH_1:=\bH_{\le w_{1b}}$, $\bH_2:=\bH_{\le w_{2a}}$,
$\bH_3:=\bH_{\ge w_{3b}}$, and $\bH_4:=\bH_{\ge w_{4a}}$. If a vertex $w_i$
does not exist, then the corresponding halfplane $\bH_i$ is not defined.
From the definition, 
the inside of $\bR$ does not
contain gates of vertices of $X$.
We denote by $\bS_i$, $i = 1,...,4$, the intersection of $\bI$ with the closure
of the complementary halfspace of $\bH_i$. We call $\bS_i$, $i=1,...4$, a
\emph{strip of $\bI$}. Consequently, the interval $\bI$ is covered by the region $\bR$, two
\emph{horizontal}  strips $\bS_1$ and $\bS_3$, and two \emph{vertical} strips
$\bS_2$ and $\bS_4$.
Using this notation, we can redefine $X_i$ as the sets of all the vertices of $X$ whose gate in $\bI$
belongs to the strip  $\bS_i$. Consequently, $X_i' \subseteq \bS_i$.
Furthermore, for a vertex $z \in \ZZ^2$, each strip $\bS_i$ is partitioned into two strips
$\bS_i'(z)$ and $\bS_i''(z)$ by the vertical or horizontal line passing via
$z$. The labeling of the strips is done in the clockwise order around $z$, 
see Fig.~\ref{fig:LSCS_cube-free_median} (left). Let $\alpha_3(X):=(s_1, t_1,
s_2, t_2, s_3, t_3, s_4, t_4)$, where
\begin{itemize}
    \item $s_1$ is a vertex of $X^+$ furthest from $x$, whose gate $s_1'$
    belongs to $\bS_1'(x)$, and $t_1$ is a vertex of $X^+$ such that its gate $t_1'$
    belongs to $\bS_1''(x)$ and the abscissa of $t_1'$ is closest to $x_a$;
    \item $s_2$ is a vertex of $X^+$ furthest from $x$, whose gate $s_2'$
    belongs to $\bS_2''(x)$, and $t_2$ is a vertex of $X^+$ such that its gate $t_2'$
    belongs to $\bS_2'(x)$ and the ordinate of $t_2'$ is closest to $x_b$;
    \item $s_3$ is a vertex of $X^+$ furthest from $x$, whose gate $s_3'$
    belongs to $\bS_3'(x)$, and $t_3$ is a vertex of $X^+$ such that its gate $t_3'$
    belongs to $\bS_3''(x)$ and the abscissa of $t_3'$ is closest to $x_a$;
    \item $s_4$ is a vertex of $X^+$ furthest from $x$, whose gate $s_4'$
    belongs to $\bS_4''(x)$, and $t_4$ is a vertex of $X^+$ such that its gate $t_4'$
    belongs to $\bS_4'(x)$ and the ordinate of $t_4'$ is closest to $x_b$.
\end{itemize}
Let $\alpha_4(X):=(p_1, q_1, p_2, q_2, p_3, q_3, p_4, q_4)$, where $p_i$ is a vertex of $X^-$ closest to $x$, whose gate $p_i'$ belongs to 
$\bS_i'(x)$, and $q_i$ is a vertex of $X^-$ closest to $x$, whose gate $q_i'$
belongs to 
$\bS_i''(x)$. 
If any of the vertices of the four groups is not defined, then its corresponding
coordinate in $\alpha(X)$ is set to $*$.
\subsection*{The reconstructor $\beta(Y)$}
Let $Y$ be a vector of $22$ coordinates corresponding to a realizable sample and grouped into four parts $Y_1 := (y_1,
y_2)$, $Y_2 := (y_3, y_4, y_5, y_6)$, $Y_3 := (y_7, y_8, y_9, y_{10}, y_{11},
y_{12}, y_{13}, y_{14})$, and
$Y_4 := (y_{15}, y_{16}, y_{17}, y_{18}, y_{19}, y_{20}, y_{21}, y_{22})$.
The reconstructor $\beta(Y)$ returns a ball $B_{r''}(y)$ by performing the
following steps (see Fig.~\ref{fig:LSCS_cube-free_median} (right)):
\begin{enumerate}
    \item Using $Y_1$, 
    canonically isometrically embed
    $I(y_1,y_2)$ into $\ZZ^2$ as $\bI$. 
    \item Using $Y_2$, 
    compute the gates $y_i'$ of $y_i$
    in $\bI$ and compute the region
    $\bR$ 
    as the intersection of the halfplanes
    $\bH_{\le y_{1b}}$, $\bH_{\le y_{2a}}$, $\bH_{\ge y_{3b}}$, and $\bH_{\ge
    y_{4a}}$ with $\bI$.
   \item Using $Y_3$, 
    compute the set
    $\bR'\subseteq \bR$ of all $y =(y_a, y_b) \in \bR$ such that the
    gates $y_7', y_8', y_9', y_{10}', y_{11}', y_{12}', y_{13}',y_{14}'$
    belong to $\bS_1'(y)$, $\bS_1''(y)$, $\bS_2''(y)$, $\bS_2'(y)$,
    $\bS_3'(y)$, $\bS_3''(y)$, $\bS_4''(y)$, $\bS_4'(y)$, resp.
    For each $y \in \bR'$, let $r_y'$ be the \emph{smallest} radius such that $Y_1
    \cup Y_3 \subseteq B_{r_y'}(y)$.
    \item Using $Y_4$, 
    compute the
    region $\bR'' \subseteq \bR$ consisting of all the vertices $y \in \bR$ such
    that the gates $y_{15}', y_{16}', \ldots, y_{21}', y_{22}'$ belong to the
    strips $\bS_1'(y)$, $\bS_1''(y)$, $\ldots$, $\bS_4'(y)$, $\bS_4''(y)$,
    respectively.
    For each $y \in \bR''$, let $r_y''$ be the \emph{largest} radius such that
    $B_{r_y''}(y) \cap Y_4 = \varnothing$.
    \item Let $\bR_0 := \{y \in \bR' \cap \bR'' : r_y'' \ge r_y' \}$ and  return as $\beta(Y)$
    any ball $B_{r_y''}(y)$ with $y \in \bR_0$.
\end{enumerate}
%
\begin{proposition}
    \label{prop:LSCS_cube-free_median}
    For any cube-free median graph $G$, the pair $(\alpha, \beta)$ of vectors defines a proper labeled sample compression scheme of size $22$ for $\mathcal B = \mathcal B(G)$.
\end{proposition}

\begin{proof}
    Let $X$ be a realizable sample for $\mathcal B$, and let $B_r(x)$ be a ball
    realizing $X$. Let $Y := \alpha(X)$.
    We suppose that $|X^+| \ge 2$, as otherwise we are trivially done. By the
    definition of the maps $\alpha$ and $\beta$, we deduce that $x$ belongs
    to the regions $\bR'$ and $\bR''$ defined in Steps 3 and 4. Also, $r_x' \le
    r \le r_x''$, and therefore, $x \in \bR_0$. This implies that $\bR_0 \ne
    \varnothing$ and the map $\beta(Y)$ is well-defined.
    Denote by $\bS_1$, $\bS_2$, $\bS_3$, and $\bS_4$ the strips defined by the
    gates $y_3'$, $y_4'$, $y_5'$, and $y_6'$ of the vertices of $Y_2$.
    Since $Y_2 = \alpha_2(X)$, all the vertices of $X$ have their
    gates in the union $\bigcup_{i=1}^4 \bS_i$ of the four strips.
    Now, suppose that $\beta(Y) = B_{r_y''}(y)$.

    \begin{claim} \label{claimX^+}
        For any $y \in \bR'$,  $X^+ \subseteq B_{r'_y}(y)$.
    \end{claim}

    \begin{proof}
        Pick any $z \in X^+$.  Assume, by symmetry, that its gate $z'$ in
        $\bI$ belongs to the strip $\bS_1$.

        \smallskip\noindent
        \textbf{Case 1.} $z' \in \bS_1''(y)$.

        In this case, $d(u^+,z) = d(u^+,z') + d(z',z)$ and $d(u^+,v^+) = d(u^+,z') +
        d(z',v^+)$. Also, $d(y,z) = d(y,z') + d(z',z)$ and $d(y,v^+) = d(y,z') +
        d(z',v^+)$. If $z \notin B_{r_y'}(y)$, then since  $v^+ \in B_{r_y'}(y)$, we
        conclude that $d(z',z) > d(z',v^+)$, yielding $d(u^+,z) > d(u^+,v^+)$, a
        contradiction with the choice of $\{u^+,v^+\}$ as a diametral pair of $X^+$.

        \smallskip\noindent
        \textbf{Case 2.} $z' \in \bS_1'(y)$.

        \smallskip\noindent
        \textbf{Subcase 2.1.} $z' \in \bS_1''(x)$.

        In this case, the vertex $y_8 \in \bS_1''(y)$ exists and $y_{8a}' \le
        z_a'$. Since $z \in \bS_1''(x)$, we have $y_a \ge z_a'$. On the other
        hand, by the construction of $\bR'$, $y_a \le y_{8a}'$. Consequently,
        $y_a \le y_{8a}' \le z_a' \le y_a$, showing that $y_a = y_{8a}' = z_a'$.
        Consequently,  $z' \in \bS_1''(y)$, which is the condition of Case 1.

        \smallskip\noindent
        \textbf{Subcase 2.2.} $z' \in \bS_1'(x)$.

        In this case, the vertex $y_7$ exists. By its definition, $y_7$ belongs
        to $\bS_1'(x)$, and by the construction of $\bR'$, $y_7$ belongs to
        $\bS_1'(y)$. If $y_a \le x_a$, then let $t$ be the vertex on the upper
        side of $\bR$ such that $t_a = y_a$. Analogously, if $y_a \ge x_a$, then
        let $t$ be the vertex on the upper side of $\bR$ such that $t_a = x_a$.
        Notice that, in both cases, we have $t \in I(x,z) \cap I(x,y_7)$ and $t \in
        I(y,z) \cap I(y,y_7)$. From the definition of the vertex $y_7$, we have
        $d(x,y_7) \ge d(x,z)$, yielding $d(t,y_7) \ge d(t,z)$. As a result, we
        obtain that $z \in B_{r_y'}(y)$.
    \end{proof}

    \begin{claim} \label{claimX^-}
        For any $y \in \bR''$, $X^- \cap B_{r''_y}(y) = \varnothing$.
    \end{claim}

    \begin{proof}
        Let $z$ be any vertex of $X^-$. Suppose, by symmetry, that
        its gate $z'$ belongs to $\bS_1'(x)$. This implies that the vertex
        $y_{15}$ exists. By the definition of $y_{15}$ and the construction of
        $\bR''$, we have $y_{15} \in \bS_1'(x) \cap \bS_1'(y)$. We distinguish
        two cases.

        \smallskip\noindent
        \textbf{Case 1.} $z' \in \bS_1'(y)$.

        If $y_a \le x_a$, then let $t$ be the vertex on the upper side of $\bR$
        such that $t_a = y_a$. Analogously, if $y_a \ge x_a$, then let $t$ be
        the vertex on the upper side of $\bR$ such that $t_a = x_a$. In both
        cases we have $t \in I(y,z) \cap I(y,y_{15})$ and $t \in
        I(x,z) \cap I(x,y_{15})$.
        Since $y_{15}$ is the vertex of $X^-$ closest to $x$ whose gate is in
        $\bS_1'(x)$, we must have $d(x,y_{15}) \le d(x,z)$. This implies that $d(t, y_{15})
        \le d(t,z)$. Consequently, $d(y,y_{15}) \le d(y,z)$, and thus, $z \notin
        B_{r_y''}(y)$.

        \smallskip\noindent
        \textbf{Case 2.} $z' \in \bS_1''(y)$.

        In this case, the vertices $z'$ and $y_{15}'$ are to the left of $x$, $y_{15}'$ is to
        the left of $y$, and $z'$ is to the right of $y$. Let $t$ be the vertex
        on the upper side of $\bR$ such that $t_a = z_a'$. Then, $t \in
        I(x,y_{15}) \cap I(x,z)$. From the choice of $y_{15}$ and this
        inclusion, we deduce that $d(t,y_{15}) \le d(t,z)$. Let $s$ be the
        vertex on the upper side of $\bR$ such that $s_a = y_a$. Notice that
        $t \in I(s,z)$ and $s \in I(t,y_{15})$. Since $d(t,y_{15}) \le
        d(t,z)$, from the previous inclusion, we obtain $d(s,y_{15}) \le d(t,
        y_{15}) \le d(t,z) \le d(s,z)$. Since $s \in I(y,y_{15}) \cap I(y,z)$,
        we conclude that $d(y,y_{15}) \le d(y,z)$.
    \end{proof}
Since $r''_y\ge r'_y$ for any vertex $y\in \bR_0=\bR'\cap \bR''$, by Claims~\ref{claimX^+} and
\ref{claimX^-}, we deduce that the ball $B_{r''}(y)$ realizes the sample $X$, concluding the proof.
\end{proof}

\section{Interval Graphs}\label{sec:interval}

\emph{Interval graphs} are the intersection graphs of sets of line segments. Interval graphs are classical in graph theory.
For any interval graph $G=(V,E)$, we construct proper LSCS 
of size 4 for $\B(G)$ and $\B_r(G)$.
We consider a representation of $G$ by a set of segments $J_v, v\in V$ of ${\mathbb R}$
with pairwise distinct ends. For any $u\in V$, we denote by $J_u=[s_u,e_u]$ its segment, where $s_u$ is the start of $J_u$, and $e_u$ is the end of $J_u$, {\it i.e.}, $s_u\leq e_u$.
We use the following property of  interval graphs:

\begin{lemma}\label{lem:intervalinclusion}
If $u,v\in B_r(x)$, $s_u,s_z<s_v$, and $e_u< e_v,e_z$, then $z\in B_r(x)$.
\end{lemma}

\begin{proof}
Since $s_z<s_v$ and $e_u< e_z$, if $J_u$ and $J_v$ intersect, then $J_z$ covers the segment $[s_v,e_u]$, and otherwise, $J_z$ intersects $[e_u, s_v]$. Let $P$ be a path obtained from a shortest $(x,u)$-path of $G$  by removing $u$, and $Q$ be a path obtained from a shortest $(x,v)$-path by removing $v$.
The union $J_S$ of all segments of $S:=P\cup \{x\} \cup Q$ intersects $J_u$ and $J_v$. If $J_u$ and $J_v$ intersect, then $J_z$ covers $[s_v,e_u]$, and thus, intersects $J_S$. Otherwise, $J_S$ covers $[e_u, s_v]$, and $J_z$ intersects $[e_u,s_v]$. In both cases, $J_z$ and $J_S$ intersect, whence a segment of $S$ intersects $J_z$. Since all segments of $S$ are at distance at most $r-1$ from $x$, $z\in B_r(x)$.
\end{proof}

Let $X$ be a realizable sample for $\B(G)$. A \emph{farthest pair} of $X^+$ is a pair $\{ u^+,v^+\}$ such that $u^+$ is the vertex in $X^+$ whose segment $J_{u^+}$ ends farthest to the left, and $v^+$ is the vertex in $X^+$ whose segment $J_{v^+}$ begins farthest to the right, {\it i.e.},
for any $w\in X^+$, we have $e_{u^+}<e_{w}$ and $s_{w}<s_{v^+}$. If $u^+\neq v^+$, then $[e_{u^+},s_{v^+}]\cap J_w\neq \varnothing$ for any $w\in X^+$.  If
$u^+=v^+$, then  $J_{u^+}\subseteq J_w$ for any $w\in X^+$. 
A vertex $p^-$ of $X^-$ is a \emph{left-bounder} if there is a ball $B_r(x)$ realizing $X$ such that $e_{p^-}<s_x$ and, for all $p\in X^-$ with $e_p<s_x$, it holds that $e_p\le e_{p^-}$. Analogously, a vertex $q^-$ of $X^-$ is a \emph{right-bounder} if there is a ball $B_r(x)$ realizing $X$ such that $e_x<s_{q^-}$ and, for all $q\in X^-$ with $e_x<s_q$, it holds that $s_{q^-}\le s_q$. If $p^-$ is a left-bounder  and $q^-$ is a right-bounder, then $\{p^-,q^-\}$ is a \emph{bounding pair} of $X^-$. The farthest pair $\{u^+,v^+\}$ of $X^+$ and the bounding pair $\{p^-,q^-\}$ of $X^-$ have the following properties:


\begin{lemma}\label{interval_graph_plus}
If $u^+,v^+\in B_r(x)$ and $r>0$, then $X^+\subseteq B_r(x)$.
\end{lemma}

\begin{proof} Pick any $w\in X^+\setminus \{u^+,v^+\}$. By the definition of $u^+$ and $v^+$, we have $s_w<s_{v^+}$ and $e_{u^+}<e_w$. If $u^+\neq v^+$, then $s_{u^+}<s_{v^+}$ and $e_{u^+}<e_{v^+}$, and so, $s_{u^+},s_w<s_{v^+}$ and $e_{u^+}<e_{v^+},e_w$. By Lemma~\ref{lem:intervalinclusion}, $w\in B_r(x)$. Now, let $u^+=v^+$.
Then, $J_{u^+}\subset J_w$, and thus, any segment intersecting $J_{u^+}$ also intersects $J_w$. Consequently, $w$ is included in any ball of $G$ of radius $r>0$ containing $u^+$, and,
in particular, $w\in B_r(x)$.  
\end{proof}

\begin{lemma}\label{interval_graph_minus}
If $e_{p^-}<s_x$ and $p^-\notin B_r(x)$, then, for all $z\in X^-$ with $e_z<e_{p^-}$, $z\notin B_r(x)$. Also, if $e_x<s_{q^-}$ and $q^-\notin B_r(x)$, then, for all $w\in X^-$ with $s_{q^-}<s_w$, $w\notin B_r(x)$.
\end{lemma}

\begin{proof}
For the first statement, towards a contradiction, suppose that $e_{p^-}<s_x$ and $p^-\notin B_r(x)$, but there exists $z\in X^-$ such that $e_z<e_{p^-}$ and $z\in B_r(x)$. Then, $s_z,s_{p^-}<s_x$ since $s_z\le e_z<e_{p^-}<s_x$, and $e_z<e_x,e_{p^-}$ as $e_z<e_{p^-}<s_x\leq e_x$. By Lemma~\ref{lem:intervalinclusion}, $p^-\in B_r(x)$, a contradiction. 
For the second statement, suppose by way of contradiction that $e_x<s_{q^-}$ and $q^-\notin B_r(x)$, but there exists $w\in X^-$ such that $s_{q^-}<s_w$ and $w\in B_r(x)$. Then, $s_x,s_{q^-}<s_w$ since $s_x\leq e_x<s_{q^-}<s_w$, and $e_x<e_w,e_{q^-}$ as $e_x<s_{q^-}<s_w\leq e_w$. By Lemma~\ref{lem:intervalinclusion}, $q^-\in B_r(x)$, a contradiction.
\end{proof}

\subsection*{The compressor $\alpha(X)$} The compressor $\alpha(X)$ of $X$ is a vector with four coordinates grouped into two pairs: $\alpha(X):=(\alpha_1(X),\alpha_2(X))$. The pair $\alpha_1(X)$ is a farthest pair $\{u^+,v^+\}$ of $X^+$ and the pair $\alpha_2(X)$ is a bounding pair $\{p^-,q^-\}$ of $X^-$. We use the symbol $*$ to indicate that the respective coordinate of $\alpha(X)$ is empty. We define $\alpha(X)$ as follows:
\begin{enumerate}
\item[(C1)] if $X^+=\varnothing$, then set $\alpha_1(X) = \alpha_2(X) :=(*,*)$;
\item[(C2)] if $X^+=\{x\}$, then set $\alpha_1(X):=(x,*)$ and $\alpha_2(X) := (*,*)$;
\item[(C3)] if $|X^+|\geq 2$, then set $\alpha_1(X):=(u^+,v^+)$ if $u^+\neq v^+$ and $\alpha_1(X):=(*,v^+)$ if $u^+=v^+$; 
	\begin{enumerate}
    \item[(C3i)] if $X^-=\varnothing$, then set $\alpha_2(X):=(*,*)$;
    \item[(C3ii)] if there exists a bounding pair of $X^-$, then set $\alpha_2(X):=(p^-,q^-)$;
    \item[(C3iii)] if there exists a left-bounder, but not a right-bounder of $X^-$, then set $\alpha_2(X):=(p^-,*)$;
    \item[(C3iv)] if there exists a right-bounder, but not a left-bounder of $X^-$, then set $\alpha_2(X):=(*,q^-)$.
    \end{enumerate}
\end{enumerate}

\subsection*{The reconstructor $\beta(Y)$} The reconstructor $\beta$ takes any signed vector $Y$ on four coordinates grouped into two pairs $Y_1$ and $Y_2$ from $\Ima(\alpha)$, and returns a ball $\beta(Y)$ defined as follows:
\begin{enumerate}
\item[(R1)] if $Y_1 = Y_2 = (*,*) $, then $\beta(Y)$ is the empty ball;
\item[(R2)] if $Y_1 =(y_1,*)$ and $Y_2 = (*,*)$, then $\beta(Y)$ is the ball of radius 0 centered at $y_1$;
\item[(R3)] if $Y_1=(y_1,y_2)$ or $Y_1=(*,y_2)$, then $\beta(Y)$ is any ball $B_r(x)$ of radius $r\ge 1$
containing $Y_1$, not intersecting $Y_2$, and such that:
\begin{enumerate}
\item[(R3i)] if $Y_2=(*,*)$, then no condition; 
\item[(R3ii)] if $Y_2=(y_3,y_4)$, then 
    $e_{y_3}<s_x$ and $e_x<s_{y_4}$;
\item[(R3iii)] if $Y_2=(y_3,*)$, then 
$e_{y_3}<s_x$;
\item[(R3iv)] if $Y_2=(*,y_4)$, then 
$e_x<s_{y_4}$.
\end{enumerate}
\end{enumerate}


Now, let $X$ be a realizable sample for $\B_r(G)$. If $|X^+|\geq 2$ or $r\ge 1$, then we define $\alpha$ and $\beta$ as above, since, in these cases, we do not specify the radius of the ball realizing $X$ in $\alpha$, nor the radius of the ball returned by $\beta$. So, we can exhibit a proper LSCS of size $4$ for $\B_r(G)$ if we can deal with the case $|X^+|\leq 1$. 
The only difference is that if $|X^+|\le 1$, then we set $\alpha_2(X)$ as in Case (C$3$), but we set $\alpha_1(X):=(*,*)$ when $X^+=\varnothing$, and $\alpha_1(X):=(*,x)$ when $X^+=\{x\}$. Now, let $r=0$. If $|X^+|=0$ and there is a ball $B_0(y)$ such that $y\notin  X^-$ and $e_y<e_z$ for any $z\in V, z\ne y$, then $\alpha(X):=((*,*),(*,*))$. Otherwise, if $|X^+|=0$, there is a ball $B_0(y)$ such that $y\notin X^-$, $w'\in X^-$, $e_{w'}<e_{y}$, and, for all $w\in V$ with $e_w<e_y$, we have $e_w\leq e_{w'}$. In this case, $\alpha(X):=((*,*),(w',*))$. If $X^+=\{x\}$, set $\alpha(X):=((x,*),(*,*))$. 
Given any signed vector $Y$ on four coordinates, $\beta$ returns a ball $\beta(Y)$ defined as follows.
If $Y=((*,*),(*,*))$, then $\beta(Y)$ is the ball $B_0(x)$ such that $e_x< e_z$ for any $z\in V\setminus \{ x\}$. If $Y=((*,*),(y_3,*))$, then $\beta(Y)$ is the ball $B_0(x)$ such that $e_{y_3}<e_x$, and, for all $w\in V$ with $e_w<e_x$, it holds that $e_w\leq e_{y_3}$. Lastly, if $Y=((x,*),(*,*))$, then $\beta(Y)$ is the ball $B_0(x)$. 

\begin{proposition} \label{LSCS-interval-graphs}
    For any interval graph $G=(V,E)$, the pair $(\alpha,\beta)$ of vectors defines
    a proper labeled sample compression scheme of size 4 for $\B(G)$ and $\B_r(G)$.
\end{proposition}

\begin{proof}
Let $X$ be a realizable sample for $\B(G)$ (the case of $\B_r(G)$ is similar), $Y=\alpha (X)$, and $B_r(x)=\beta(Y)$. The cases (R$k$) and their subcases in the definition of $\beta$ correspond to the cases (C$k$) and their subcases in the definition of $\alpha$. The correctness is trivial if $k=1,2$. 
Now, let $k=3$. Since $Y_1$ always contains a farthest pair of $X^+$ and the returned ball $B_r(x)$ contains $Y_1$ and $r\ge 1$, by Lemma~\ref{interval_graph_plus}, $X^+\subseteq B_r(x)$. Furthermore, in Case (C$3$), any ball realizing $X$ must have a radius $r\geq 1$ since $|X^+|\geq 2$. 
Now, we prove that $X^-\cap B_r(x)=\varnothing$. This is trivial in subcase (R$3$i) since $X^-=\varnothing$.  In the remaining subcases of (R$3$),
$X^-\cap B_r(x)=\varnothing$ follows from the definition of the corresponding subcase of case (C$3$) and Lemma~\ref{interval_graph_minus}. 
\end{proof}

\section{Split graphs}\label{sec:split}

In this section, we study proper labeled sample compression schemes for split graphs. A \emph{split graph} is a graph $G=(V,E)$ whose set of vertices $V$ can be partitioned into a clique $S$ and an independent set $I$.
First, recall that the sample compression conjecture for ball of radius 1 of
split graphs is as hard as the general sample compression conjecture. Indeed,
it is well known that, to any concept class ${\mathcal C}$ of a set $S$, one can associate
a split graph $G$ having $S$ as a clique, and a vertex $v_C$ for each concept $C$ of $\mathcal C$.
Then, each ball $B_1(v_C)$ coincides with $C\cup \{ v_C\}$. One can easily check that the VC-dimension
of the family of balls $\{ B_1(v_C): C\in {\mathcal C}\}$ is the same as the VC-dimension of $\mathcal C$. A
sample compression scheme for  $\mathcal C$ corresponds to a sample compression scheme for $\{ B_1(v_C): C\in {\mathcal C}\}$, where
each sample $X$ has its support in $S$.

Let $G=(V,E)$ be a split graph and let
$S=\{w_1,\ldots,w_{\omega(G)}\}$ be the clique, 
where $\omega(G)$ is the clique number of $G$ (the number of vertices in a maximum clique
in $G$). For the family $\B(G)$, we construct proper labeled sample
compression schemes of size~$k=\max\{2,\omega(G)\}$ with the
assumption that $\alpha(X)$ is a vector on $k$ coordinates. Note that
the family of balls of a split graph has VC-dimension at most $k$.
We first define the compressor $\alpha$.
\subsection*{The compressor $\alpha(X)$} We use $*$ to indicate that the respective coordinate of $\alpha(X)$ is empty. For any realizable sample $X$ for $\B(G)$,

\begin{enumerate}
\item [(C1)] if $X^+=\varnothing$, then set $\alpha(X):=(*,\ldots,*)$;
\item [(C2)] otherwise, if $X^-=\varnothing$, then set
  $\alpha(X):=(u,*,\ldots,*)$, where $u$ is some vertex in $X^+$;
\item [(C3)] otherwise, if $X^+=\{u\}$, then set
  $\alpha(X):=(*,u,*,\ldots,*)$;
\item [(C4)] otherwise, if there exists $u\in X^+\cap I$ such that
  $B_1(u)\cap X=X^+$, then set $\alpha(X):=(u,v,*,\ldots,*)$, where
  $v\in X^+\cap N(u)$;
\item [(C5)] otherwise, $k=\omega(G)$, and, if there exists
  $u\in I\setminus X^+$ such that $B_1(u)\cap X=X^+$, then, for each
  $1\leq i\leq k$, set the $i^{th}$ coordinate of $\alpha(X)$ to $w_i$
  if $w_i\in X$, and $*$ otherwise;
\item [(C6)] otherwise, if there exists $u\in S$ such that
  $B_1(u)\cap X=X^+$, then, for each $w_i\in S$,
  \begin{enumerate}
  \item if there exists $y\in N(w_i)\cap X^-$ such that
    $y\notin \alpha(X)$, then set $y$ as the $i^{th}$ coordinate of
    $\alpha(X)$;
  \item otherwise, if there exists $z\in X^+\setminus N(w_i)$ and
    $z\notin \alpha(X)$, then set $z$ as the $i^{th}$ coordinate of
    $\alpha(X)$;
  \item otherwise, set $*$ as the $i^{th}$ coordinate of $\alpha(X)$.
  \end{enumerate}
\item [(C7)] otherwise, there exists $u\in I$ such that
  $B_2(u)\cap X=X^+$, and, for each $w_i\in S$,
  \begin{enumerate}
  \item if there exists $y\in N(w_i)\cap X^-$ such that
    $y\notin \alpha(X)$, then set $y$ as the $i^{th}$ coordinate of
    $\alpha(X)$;
  \item otherwise, if $w_i\in N(u)$ and there exists
    $z\in X^+\cap N(w_i)$ such that $z\in I$ and $z\notin \alpha(X)$,
    then set $z$ as the $i^{th}$ coordinate of $\alpha(X)$;
  \item otherwise, set $*$ as the $i^{th}$ coordinate of $\alpha(X)$.
  \end{enumerate}
\end{enumerate}

\subsection*{The reconstructor $\beta(Y)$} The reconstructor $\beta$ takes any sign vector $Y$ on $k$ coordinates and returns a ball defined in the following way:

\begin{enumerate}
\item [(R1)] if $Y=(*,\ldots,*)$, then $\beta(Y)$ is the empty ball;
\item [(R2)] if $Y=(u,*,\ldots,*)$ for some vertex $u\in Y^+$, then
  $\beta(Y)$ is any ball of $G$ covering $G$;
\item [(R3)] if $Y=(*,u,*,\ldots,*)$ for some vertex $u\in Y^+$, then
  $\beta(Y)=B_0(u)$;
\item [(R4)] if $Y=(u,v,*,\ldots,*)$ for two vertices $u,v\in Y^+$
  such that $u\in I$ and $v\in S$, then $\beta(Y)=B_1(u)$;
\item [(R5)] if there exist two vertices $x,y\in Y^+$ such that
  $x,y\in S$, then $\beta(Y)$ is any ball of radius~$1$ centered at a
  vertex $u\in I$ such that $B_1(u)\cap Y=Y^+$;
\item [(R6)] if there exists $x\in Y^-$ and, for all $y_j\in Y^+$ such
  that $y_j$ is the $j^{th}$ coordinate of $Y$, $y_j\notin N(w_j)$,
  then $\beta(Y)$ is any ball of radius~$1$ centered at a vertex
  $w_t\in S$ such that the $t^{th}$ coordinate of $Y$ is~$*$ and
  $B_1(w_t)\cap Y=Y^+$;
\item [(R7)] otherwise, $\beta(Y)$ is any ball $B$ of radius~$2$
  centered at a vertex $u\in I$ such that $B\cap Y^-=\varnothing$,
  and, for all $y_t\in Y^+$ such that $y_t$ is the $t^{th}$ coordinate
  of $Y$, $w_t\in N(u)$.
\end{enumerate}

\begin{proposition}
  For any split graph $G$, the pair $(\alpha,\beta)$ of vectors
  defines a proper labeled sample compression scheme of
  size~$\omega(G)$ for $\B(G)$.
\end{proposition}

\begin{proof}
  Let $X$ be a realizable sample for $\B(G)$. We will prove that the
  ball $\beta(Y)$ realizes the sample $X$, {\it i.e.}, that
  $\beta(Y)\cap X=X^+$. For each $1\leq x\leq 7$, the Case (R$x$) in
  the definition of $\beta$ corresponds to the Case (C$x$) in the
  definition of $\alpha$, and the correctness follows as long as no
  two cases (C$i$) and (C$j$) produce the same type of vector. First,
  note that, in (C$5$), $\alpha(X)$ contains at least two vertices of
  $X^+$ that are in $S$ since $|X^+|\geq 2$ and $B_1(u)$ realizes $X$
  for some vertex $u\in I$ such that $u\notin X^+$. Also, by
  definition, in (C$6$) and (C$7$), $\alpha(X)$ does not contain any
  vertex of $X^+$ that is in $S$, and contains at least one vertex of
  $X^-$ that is in $I$ since $|X^-|\geq 1$ and any ball realizing $X$
  contains all of $S$ in both of these cases. Lastly, in (C$7$), since
  we are not in any of the other cases, there must be at least one
  vertex in $X^+\cap I$, and thus, for some $1\leq j\leq k$, the
  $j^{th}$ coordinate of $\alpha(X)$ contains a vertex
  $y\in X^+\cap N(w_j)$, while this is never the case in
  (C$6$). Hence, no two cases produce the same vector, and the
  correctness follows.
\end{proof}

\section{Planar graphs}\label{sec:planar}

In this section, we study proper labeled sample compression schemes
for planar graphs. There seems to be an inherent difficulty when
dealing with planar graphs, and thus, we only give partial results in
this section. Let $G$ be a planar graph and fix a planar embedding of
$G$. For the family $\B_1(G)$, we construct proper labeled sample
compression schemes of size~$4$. Note that the family of balls of a
planar graph has VC-dimension at most~$4$ \cite{BouTh,ChEsVa}. For a
vertex $x$, we denote by $N(x)$ its open neighborhood, {\it i.e.},
$N(x)=B_1(x)\setminus \{ x\}$. We begin with the next simple lemma,
which is a direct consequence of the fact that a planar graph does not
admit $K_{3,3}$ as a minor. Let $X$ be a realizable sample. A
\emph{potential center} for $X$ is a vertex $v$ of $G$ such that
$X^+\subseteq N(v)$ and $v\notin X^+$.

\begin{lemma} \label{planar-potential-center} If $|X^+|\ge 3$ and, for any vertex $v\in X^+$, it holds that $N(v)\cap X^-\neq \varnothing$, then $X$ either has one or two potential centers.
\end{lemma}

\begin{proof} Since $X$ is realizable, there exists $u\in V$ such that $u\notin X^+$, $X^+\subseteq N(u)$, and $X^-\cap N(u)=\varnothing$, and thus, $X$ has at least
one potential center (namely $u$). If $X$ has three or more potential centers, then since $|X^+|\ge 3$, we obtain a forbidden $K_{3,3}$ (see Fig.~\ref{fig:K33_planar_potential_centers} for an illustration).
\end{proof}

\begin{figure}[htb]
    \centering
    \includegraphics[width=0.28\linewidth]{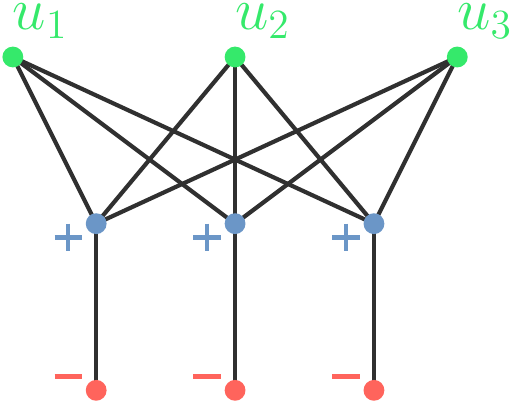}
    \caption{
        \label{fig:K33_planar_potential_centers}
        Illustration of the forbidden $K_{3,3}$ obtained in the proof of Lemma~\ref{planar-potential-center} when there are three potential centers $u_1$, $u_2$, and $u_3$.
    }
\end{figure}

We can now describe our proper labeled sample compression scheme of size~$4$. The compressor $\alpha(X)$ is a vector on four coordinates. We use the symbol $*$ to indicate that the respective coordinate of $\alpha(X)$ is empty. If $|X^+|=0$, then $\alpha(X):=(*,*,*,*)$, and $\beta((*,*,*,*))$ is the empty ball. If the center of a ball realizing $X$ is contained in $X^+$, say $u$, then $\alpha(X):=(u,*,*,*)$, and $\beta((u,*,*,*))$ is $B_1(u)$. Otherwise, if $|X^+|\ge 3$, then the compressor $\alpha(X)$ is defined as follows. First, let $\alpha(X)$ be equal to $(v_1,v_2,v_3,*)$, where $v_1,v_2,v_3$ are
three arbitrary vertices of $X^+$. If $X^-$ contains a vertex $z$ adjacent
to a potential center of $X$, then change the last coordinate of $\alpha(X)$ to $z$, {\it i.e.}, set $\alpha(X):=(v_1,v_2,v_3,z)$. If $Y$ is a signed vector with $|Y^+|=3$, then the reconstructor
$\beta(Y)$ is defined as follows. If $|Y^-|=0$, then $\beta(Y)$ is $B_1(u)$ for any vertex $u$ adjacent to the three vertices of $Y^+$. If $|Y^-|=1$, then
$\beta(Y)$ is $B_1(u)$ for any vertex $u$ adjacent to the three vertices of $Y^+$, and not adjacent to the unique vertex of $Y^-$.


Now, suppose that $|X^+|\le 2$ and that the center of a ball realizing $X$ is not contained in $X^+$. In this case, the number of potential centers is potentially unbounded,
and thus, we need a technique different from the one described above. We first
deal with the case $|X^+|=2$.
For any two vertices $u,v\in V$, let $\ell_{u,v}$ be the labeling
that assigns labels $0$ through $k=|N(u)\cap N(v)|-1$ to the vertices of
$N(u)\cap N(v)$. In particular, let $\ell_{u,v}$ have the property that, for
all $0\leq i < j\leq k$, if $w_i, w_j\in N(u)\cap N(v)$ are such that
$\ell_{u,v}(w_i) = i$ and $\ell_{u,v}(w_j) = j$, then, for all $i<y<j$, the
vertex $w_y$ with label $y$ according to $\ell_{u,v}$ is contained in the
interior region of the cycle $u,w_i,v,w_j,u$ (see Fig.~\ref{fig:l_uv_planar_graph} for an illustration).
Such a labeling always exists since the subgraph induced by the vertices $u,v$, and $w_1,\ldots,w_k$ is planar and contains $K_{2,k}$ as a subgraph.

\begin{figure}[htb]
    \centering
    \includegraphics[width=0.28\linewidth]{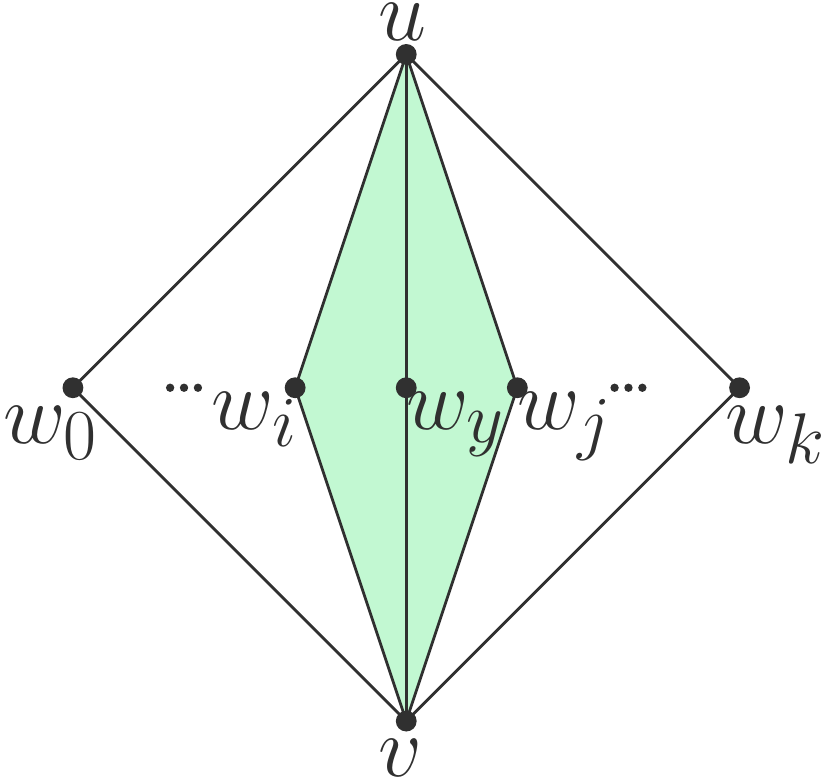}
    \caption{
        \label{fig:l_uv_planar_graph}
        Illustration of the labeling $\ell_{u,v}$ described in the case of planar graphs when $X^+=\{u,v\}$, $W=\{w_1,\dots,w_k\}$ is the set of the potential centers of $X$, where, for all $0\leq i\leq k$, $w_i$ has label $i$ according to $\ell_{u,v}$, and $i<y<j$.
    }
\end{figure}

Now, let $X$ be a realizable sample with $X^+=\{u,v\}$, and let $W=\{w_0,\ldots,w_k\}$ be the set of its potential centers, where, for all $0\leq i\leq k$, $w_i$ has label $i$ according to $\ell_{u,v}$. Since $X$ is realizable, $k\geq 0$, and in what follows, all indices are considered modulo~$k+1$. Then, as a consequence of the property of $\ell_{u,v}$ described in the paragraph above, we get the following:

\begin{lemma}\label{consecutive-centers}
For every vertex $x\in X^-$, one of the following holds:

\begin{enumerate}
\item $B_1(x)\cap W=\varnothing$;
\item $B_1(x)\cap W= \{w_j\}$, for some $0\leq j\leq k$;
\item $B_1(x)\cap W= \{w_j,w_{j+1}\}$, for some $0\leq j\leq k$;
\item $B_1(x)\cap W= \{w_j,w_{j+1},w_{j+2}\}$, for some $0\leq j\leq k$.
\end{enumerate}
\end{lemma}

The compressor $\alpha(X)$ is defined as follows. First, let $\alpha(X)=(u,v,*,*)$ and note that $\alpha^+(X):=X^+$. If $X^-$ does not contain a vertex that is
a potential center nor a vertex adjacent to a potential center, then change nothing, {\it i.e.}, $\alpha^-(X):=\varnothing$. Otherwise, there exists some
$0\leq s\leq k$ such that $B_1(w_s)\cap X=X^+$, and some vertex
$t\in X^-$ such that $w_{s-1}\in B_1(t)$, and so, change the third coordinate of $\alpha(X)$ to $t$, and note that
$\alpha^-(X)=(t)$. If $Y$ is a signed map with $Y^+=(u,v)$, then
the reconstructor $\beta(Y)$ is defined as follows. If $|Y^-|=0$, then
$\beta(Y)=B_1(w)$ for any vertex $w$ adjacent to both the vertices of
$Y^+$. If $Y^-=(t)$, then $\beta(Y)$ is $B_1(w_s)$, where the vertices
$w_s$ and $w_{s-1}$ are labeled $s$ and $s-1$, respectively, according
to $\ell_{u,v}$, and $w_{s-1}\in B_1(t)$ while $w_{s}\notin B_1(t)$.

Now, let us finish with the case $|X^+|=1$. For any vertex $u\in V$,
let $\ell_u$ be the labeling that assigns labels $0$ through $k=|N(u)|-1$
to the vertices of $N(u)$ according to the circular order they appear
in topologically around $u$ in the clockwise direction, and starting
at any vertex of $N(u)$. For a subset $S$ of vertices of $N(u)$, the
{\it predecessor} (\emph{successor}, resp.) of a vertex $x\in S$, is
the first vertex of $S$ that appears in the circular order
topologically around $u$ in the counterclockwise (clockwise, resp.)
direction.
Now, let $X$ be a realizable sample with $X^+=\{u\}$. Since
$B_1(u)$ does not realize $X$ (which means that
$N(u) \cap X^- \neq \varnothing$), any ball of radius $1$ realizing
$X$ is centered at a neighbor of $u$. Consider the set
$N(u) = \{w_0,\ldots,w_{k}\}$ of the neighbors of $u$, where, for all
$0\leq i\leq k$, $w_i$ has label $i$ according to $\ell_u$. Since $u$
has a neighbor in $X^-$, there exists some $0\leq s\leq k$ such that
$B_1(w_s)\cap X=X^+$, and some vertex $t\in X^-$ such that
$w_{s-1}\in B_1(t)$ (again, here and further, the indices are considered modulo $k+1$). Let
$W = \{w_s \in N(u) : w_s \notin B_1(t) \text{ and } w_{s-1} \in B_1(t)\}$.
Our labeled sample compression scheme will return a ball of
radius $1$ centered at a vertex of $W$. Note that, in the embedding of
$G$, if $u$ and $t$ have at least two common neighbors, there exist
$0 \leq p,q \leq k$ such that $w_p$ is the successor of $w_q$ in $N(u) \cap B_1(t)$, and all
the vertices of $N(u) \cap B_1(t)$
distinct from $u$, $t$, $w_p$, and $w_q$ are in the interior region of the cycle
$u,w_p,t,w_q,u$ (where $t$ is omitted if it coincides with $w_p$ or
$w_q$). Hence, if $p<q$, then, for every vertex $w_j\in (N(u) \cap (B_1(t))\setminus \{u,t,w_p,w_q\}$,
it holds that $p<j<q$. Otherwise, if $p>q$, then, for every vertex
$w_j\in (N(u) \cap B_1(t))\setminus \{u,t,w_p,w_q\}$, either $j>p$ or $j<q$.
If $N(u) \cap B_1(t) = \{w_i\}$, then we let $p= q= i$.
Note that, for any $w_r, w_{r'} \in N(u) \cap B_1(t)$ such that $w_r$ is the
predecessor of $w_{r'}$ in $N(u) \cap B_1(t)$, if $r<r'-1$, then
there is at most one vertex $w_s\in W$ such that $r<s<r'$.
Similarly, if $r>r'$ and $w_r$ is not the predecessor of $w_{r'}$ in $N(u)$, then
there is at most one vertex $w_s\in W$ such that $s>r$ or $s<r'$.

To define $\alpha(X)$, we consider the first element $w_s$ of $W$
appearing after $w_p$ (for the vertices of $N(u)$ in the circular
order topologically around $u$ in the clockwise direction) such that
$B_1(w_s) \cap X = X^+$. If $s-1\geq p$ and, for all $p\leq j\leq s-1$,
$w_j\in B_1(t)$, we let $\alpha^-(X) := (t)$ and
$\alpha^+(X) :=(u)$. Similarly, if $s-1<p$ and, for all $j\geq p$ and $i\leq s-1$,
$w_i,w_j\in B_1(t)$, we let $\alpha^-(X) := (t)$ and $\alpha^+(X) :=(u)$.
That is, in both of the previous cases, $w_s$ is the first element of $W$ appearing
after $w_p$. Note that if $u$ and $t$ have a unique common
neighbor, then we are in the first case by the definition of $t$
(their unique common neighbor is $w_{s-1}=w_p$). Otherwise, let $w_{s'}$ be the
predecessor of $w_s$ in $W$. Observe that $w_{s'}$ belongs to the
interior region of the cycle $u,w_{s'-1},t,w_{s-1},u$ (where $t$ is
omitted if it coincides with $w_{s'-1}$ or $w_{s-1}$) because
$w_{s'-1},w_{s-1}\in B_1(t)$ since $w_{s'},w_s\in W$. Since $B_1(w_{s'})$ does
not realize $X$, there exists $z \in B_1(w_{s'}) \cap X^-$. In this case,
we let $\alpha^+(X) :=(u)$ and $\alpha^-(X) := (t,z)$ (where we
remember the order between $t$ and $z$).

\begin{claim}\label{claimzapastropdevoisins}
  Either $B_1(z) \cap W = \{w_{s'}\}$, or $B_1(z) \cap W =
  \{w_{s'},w_{s''}\}$, where $w_{s''}$ is the predecessor of $w_{s'}$
  in $W$.
\end{claim}

\begin{proof}
  If $z=w_{s'}$, then the claim follows since $w_{s'}$ belongs to the
  interior region $R$ of the cycle $u,w_{s'-1},t,w_{s-1}$, and $z \notin B_1(w_s)$.
  Hence, assume that $z$ is adjacent to $w_{s'}$. Then, either $z$ also
  belongs to $R$,
  or it coincides with one of the vertices $u$, $t$, $w_{s'-1}$ or
  $w_{s-1}$. Since $z \in X^-$ and $u \in X^+$, we have $z \neq
  u$. Since $w_{s'} \in B_1(z) \setminus B_1(t)$, we have $z \neq t$.

  If $z$ belongs to $R$, then $w_{s'}$ is the only neighbor of $z$ in
  $W$, since the only vertex of $W$ in the cycle
  $u,w_{s'-1},t,w_{s-1},u$ and its interior region is $w_{s'}$. If $z$
  coincides with $w_{s'-1}$ and $s'-1\neq p$, then let $w_r$ be the
  predecessor of $w_{s'-1}$ in $N(u) \cap B_1(t)$. Note that $r\neq q$
  since $s'-1\neq p$. Then, $z$ is in the interior region of the cycle
  $u,w_r,t,w_{s-1},u$. In this case, since the only vertices of $W$ in
  this cycle and its interior region are $w_{s'}$ and potentially
  $w_{s''}$ (if $r = s''-1$), we are done. Suppose now that
  $z = w_{s'-1} = w_p$. If $w_{s-1} = w_q$, then $W=\{w_{s'},w_s\}$
  and $B_1(z)\cap W=\{w_{s'}\}$. If $z = w_{s'-1} = w_p$ and $w_{s-1} \neq w_q$,
  then $z$ is not adjacent to any vertex lying in
  the interior region of the cycle $u,w_{s-1},t,w_q,u$.  The only
  vertices of $W$ that are not lying in this region are $w_{s'}$ and
  $w_{q+1}$ if it is in $W$. Since $w_{q+1}$ is the predecessor of
  $w_{s'}$ in $W$ in this case, we are done.

  Suppose now that $z$ coincides with $w_{s-1}$. If $s-1 \neq q$, then
  $z$ belongs to the interior region of the cycle
  $u, w_{s'-1},t,w_{r'},u$, where $w_{r'}$ is the successor of
  $w_{s-1}$ in $N(u) \cap B_1(t)$. Since the only vertices of $W$
  lying in this region are $w_{s'}$ and $w_s$, and since
  $z \notin B_1(w_s)$, $w_{s'}$ is the unique neighbor of $z$ in $W$.
  Finally, suppose that $z = w_{s-1} = w_q$. If $w_{s'-1} = w_p$, then
  as before, $W=\{w_{s'},w_s\}$ and $B_1(z)\cap W=\{w_{s'}\}$. If
  $z = w_{s-1} = w_q$ and $w_{s'-1} \neq w_p$, then $z$ is not
  adjacent to any vertex lying in the interior region of the cycle
  $u,w_p,t,w_{s'-1},u$. Since $w_{s'}$ and $w_{s}$ are the only
  vertices of $W$ that are not lying in this region, and since
  $z \notin B_1(w_s)$, we are done.
\end{proof}

We now define the reconstructor $\beta(Y)$ when $Y^+ = (u)$ and
$|Y^-| \geq 1$. If $|Y^-| = 1$, we let $Y^- = (t)$, and if
$|Y^-| = 2$, we let $Y^- = (t,z)$. We define $w_0,\dots,w_k$
(in particular, $w_p$ and $w_q$) and $W$ as before.
If $Y^- = (t)$, let $w_s$ be the first element of $W$
appearing after $w_p$ (for the vertices of $N(u)$ in the circular
order topologically around $u$ in the clockwise direction), and let
$\beta(Y) = B_1(w_s)$. Otherwise, $Y^-= (t,z)$ and $|B_1(z)\cap W|\in \{1,2\}$
by Claim~\ref{claimzapastropdevoisins}. If $|B_1(z)\cap W|=1$, then
let $w'$ be the only vertex in $B_1(z)\cap W$. Otherwise,
$|B_1(z)\cap W|=2$ and one of the two vertices is a predecessor of the other in $W$,
say $B_1(z)\cap W=\{w',w''\}$, and let $w''$ be the predecessor of $w'$ in $W$.
In both cases, $\beta(Y) = B_1(w)$, where $w$ is the successor of $w'$ in $W$.

\begin{proposition}
For any planar graph $G$, the pair $(\alpha,\beta)$ of vectors defines a proper labeled sample compression scheme of size 4 for $\B_1(G)$.
\end{proposition}

\begin{proof}
  Let $X$ be a realizable sample for $\B_1(G)$. We prove that the ball
  $\beta(Y)$ realizes the sample $X$, {\it i.e.}, that
  $\beta(Y)\cap X=X^+$. When $Y=(*,*,*,*)$, then $|X^+|=0$, and so, the
  empty set returned is compatible with $X$. When $|Y^+|=1$ and $|Y^-|=0$, this
  implies that the ball of radius $1$ centered at the unique vertex of
  $Y^+$ is compatible with $X$, and so, $\beta(Y)$ is compatible with
  $X$ in this case. When $|Y^+|=3$, the correctness follows from
  Lemma~\ref{planar-potential-center}, and when $|Y^+|=2$, the
  correctness follows from Lemma~\ref{consecutive-centers}. When
  $Y^+= (u)$ and $Y^-= (t)$ (or $Y^- = (t,z)$), then we define
  $w_0,\dots,w_k$ (in particular, $w_p$ and $w_q$) and $W$ as before.
  If $Y^-=(t)$, the ball of radius one
  centered at the first vertex in $W$ appearing after $w_p$ (for the vertices
  of $N(u)$ in the circular order topologically around $u$ in the clockwise direction)
  realizes $X$, and it is the ball returned by
  $\beta(Y)$.  If $Y^- = (t,z)$, then by the definition of $\alpha$,
  there exists $w \in W$ such that $B_1(w)\cap X = X^+$ and $z$ is
  in the ball of radius~$1$ centered at the predecessor of $w$ in $W$. Since
  $\beta(Y)$ returns $B_1(w)$ in this case, we are done.
\end{proof}

\section{Hyperbolic graphs}\label{sec:hyperbolic}

A \emph{$(\rho, \mu)$-approximate proper labeled sample compression scheme of
size $k$} for the family of balls $\B(G)$ of a graph $G$ 
compresses any realizable sample $X$ to a subsample $\alpha(X)$ of support of size $k$, such that
$\beta(\alpha(X))$ is a ball $B_{r}(x)$ such that $X^+ \subseteq B_{r+\rho}(x)$
and $X^- \cap B_{r-\mu}(x) = \varnothing$. Let $(V,d)$ be a metric space and
$w\in V$.
Let $\delta\ge 0$. A metric space $(X,d)$ is \emph{$\delta$-hyperbolic}
\cite{Gr}
if, for any four points $u,v,x,y$ of $X$, the two larger of the
sums $d(u,v)+d(x,y)$, $d(u,x)+d(v,y)$, and $d(u,y)+d(v,x)$, differ by at most
$2\delta \geq 0$. Hyperbolic metric spaces and graphs play an important role in geometric group theory,
geometry of negatively curved spaces, and have become of interest in network science.
Next, we show that $\delta$-hyperbolic graphs admit a ($2\delta,
3\delta)$-approximate labeled sample compression scheme of size $2$.

An interval $I(u,v)$ of a graph 
is \emph{$\nu$-thin} if $d(x,y) \le \nu$ for any two points $x,y \in I(u,v)$ with
$d(u,x) = d(u,y)$ and $d(v,x) = d(v,y)$. 
Intervals of $\delta$-hyperbolic 
graphs are $2\delta$-thin. A metric space  $(X,d)$ is {\it injective} if, whenever $X$ is
isometric to a subspace $Z$ of a metric space $(Y,d')$, 
there is a map $f:Y\rightarrow Z$ such that
$f(z)=z$ for any $z\in Z$ and $d'(f(x),f(y)) \le d'(x,y)$ for any $x,y\in Y$.
By a construction of Isbell~\cite{Is} (rediscovered by Dress~\cite{Dr}), any metric space $(V,d)$ has an {\it injective hull} $E(V)$, {\it i.e.}, the smallest injective metric space into which $(V,d)$ isometrically embeds.
Lang~\cite{La} proved that the injective hull of a $\delta$-hyperbolic space is
$\delta$-hyperbolic. It was shown in~\cite{Dr} that the injective hull $T :=
T(u,v,y,w)$ of a metric space on 4
points $u,v,y,w$ is a rectangle $R := R(u',v',y',w')$ with four attached tips
$uu', vv',yy',ww'$ (one or several
tips may reduce to a single point or $R$ may reduce to a
segment or a single point). The smallest side of $R$ is exactly the
hyperbolicity of the quadruplet $u,v,y,w$.

Let $X$ be a realizable sample of $\B(G)$ and $\{u^+,v^+\}$ be a diametral pair of $X^+$.
Let  $B_{r^*}(y)$ be a ball of smallest radius such that $X^+ \subseteq B_{r^*}(y)$ and $X^- \cap B_{r^*}(y) = \varnothing$.
We set $\alpha(X) := \varnothing$ if $X^+ = \varnothing$, $\alpha(X) := X^+$ if
$|X^+| = 1$, and $\alpha(X) := \{u^+, v^+\}$ if $|X^+| \ge 2$. Given a subset $Y$ of size at
most 2, the reconstructor returns $\beta(Y) = \varnothing$ if $Y = \varnothing$, $\beta(Y) = B_0(y)$
if $Y = \{y\}$, and $\beta(Y) = B_{d(y_1,y_2)/2}(x)$ if $Y = \{y_1,y_2\}$,
where $x$ is the middle of a $(y_1,y_2)$-geodesic.




\begin{proposition}
    \label{prop:LSCS_delta-hyperbolic}
    For any $\delta$-hyperbolic graph $G = (V,E)$,  the pair $(\alpha, \beta)$
    defines a $(2\delta, 3\delta)$-approximate proper labeled sample
    compression scheme of size $2$ for $\B(G)$.
\end{proposition}

\begin{proof}
    We first show that $X^+ \subseteq B_{r + 2\delta}(x)$, where $r =
    d(u^+,v^+)/2$ and $x$ is a middle of a $(u^+,v^+)$-geodesic. Pick any $w\in X^+$.
    Since $u^+,v^+$ is a diametral pair of $X^+$, $d(u^+,w)\le 2r$ and $d(v^+,w)\le 2r$. We also have $d(u^+,v^+)=2r$ and $d(x,u^+)=d(x,v^+)=r$.
    Thus, the three distance sums have the form $d(u^+,w)+d(x,v^+)\le 3r$, $d(v^+,w)+d(x,u^+)\le 3r,$ and $d(u^+,v^+)+d(x,w)=2r+d(x,w)$. By the definition of
    $\delta$-hyperbolicity, we conclude that either $d(x,w)\le r$ (if $d(u^+,v^+)+d(x,w)$ is at most $3r$)  or $d(x,w)\le r+2\delta$
    (if $d(u^+,v^+)+d(x,w)$ is the largest sum).
    Hence, $w\in B_{r+2\delta}(x)$.
We now show that $X^- \cap B_{r - 3\delta}(x) = \varnothing$. 
    Pick $w \in X^-$ and consider the injective hull $T$ of the
    points $\{u^+, v^+, y, w\}$. 
    $T$ is a rectangle $R$ with four tips (see Fig.
    \ref{fig:LSCS_delta-hyperbolic}) and is a subspace of the
    injective hull $E(V)$. Since $w \in X^-$, $w \notin B_{r^*}(y)$. Since $u^+, v^+ \in
    B_{r^*}(y)$, we deduce that $d(y,w) > d(y,u^+)$ and $d(y,w) > d(y,v^+)$.
    Let $x'$ be a point of $I(u^+,v^+) \cap T$ such that $d(u^+, x') = d(u^+,
    x) = r$ and $d(v^+, x') = d(v^+, x) = r$. Since the injective hull $T$ is
    $\delta$-hyperbolic, its intervals are $2\delta$-thin, and thus, $d(x,x') \le
    2\delta$.


    \smallskip\noindent
    \textbf{Case 1.} $u^+$, $v^+$, $y$, and $w$ are as
    in Fig.~\ref{fig:LSCS_delta-hyperbolic}(1).
    First, suppose that $x'$ belongs to the tip between $u^+$ and $u'$ or
    to the segment between $u'$ and $v'$.
    Since $y'$ and $w'$ belong to a common geodesic from $y$ to $w$ and from
    $y$ to $v^+$, and since $v^+ \in B_{r^*}(y)$ and $w \notin B_{r^*}(y)$, we
    deduce that $d(w,w') > d(w',v^+) \ge d(v',v^+)$.
    Consequently, $d(v',w) > d(v',v^+)$.
    If $x'$ is located on the tip between $u^+$ and $u'$ or on the segment
    between $u'$ and $v'$, then, since $r = d(x',v^+) = d(x',v') + d(v',v^+)$
    and
    $d(x',w) = d(x',v') + d(v',w)$, we obtain that $w \notin
    B_{r}(x')$.
    Since $d(x,x') \le 2\delta$, $w \notin B_{r - 2\delta}(x)$.
    If $x'$ belongs to the tip between $v'$ and $v^+$, then $r = d(x',v^+) \le
    d(v',v^+) \le d(v',w)$, 
    whence $w \notin
    B_{r}(x')$ and $w \notin B_{r - 2\delta}(x)$.

    \smallskip\noindent
    \textbf{Case 2.} $u^+$ and $v^+$, and $y$ and $w$ are opposite in $T$
    as in Fig.~\ref{fig:LSCS_delta-hyperbolic}(2).
    Consider $x'$ to be on the boundary of $T$ containing the vertices
    $u'$, $w'$, and $v'$.
    Since $v^+ \in B_{r^*}(y)$ and $w \notin B_{r^*}(y)$, then
    $d(v',w') + d(w',w) > d(v',v^+)$. Note also that $d(v',w') \le \delta$,
    and thus, $d(w,v') > d(v',v^+) - \delta$.
    Independently of the location of $x'$ on the boundary of $T$,
    $w\notin B_{r - \delta}(x')$.
    Thus, $w \notin B_{r - 3\delta}(x)$.

    \smallskip\noindent
    \textbf{Case 3.} $u^+$, $v^+$, $y$, and $w$ are as in
    Fig.~\ref{fig:LSCS_delta-hyperbolic}(3).
     Since $w'$ belongs to a geodesic between $y$ and $w$ and between
    $y$ and $v^+$, and $w \notin B_{r^*}(y), v^+ \in B_{r^*}(y)$, we
    deduce that $d(w',w) > d(w',v') + d(v',v^+) \ge d(v',v^+)$.
    Independently of the location of $x'$ on the boundary of $T$, we obtain that $w \notin
    B_{r - 2\delta}(x)$.
\end{proof}
\begin{figure}
    \centering
    \includegraphics[width=0.751\linewidth]{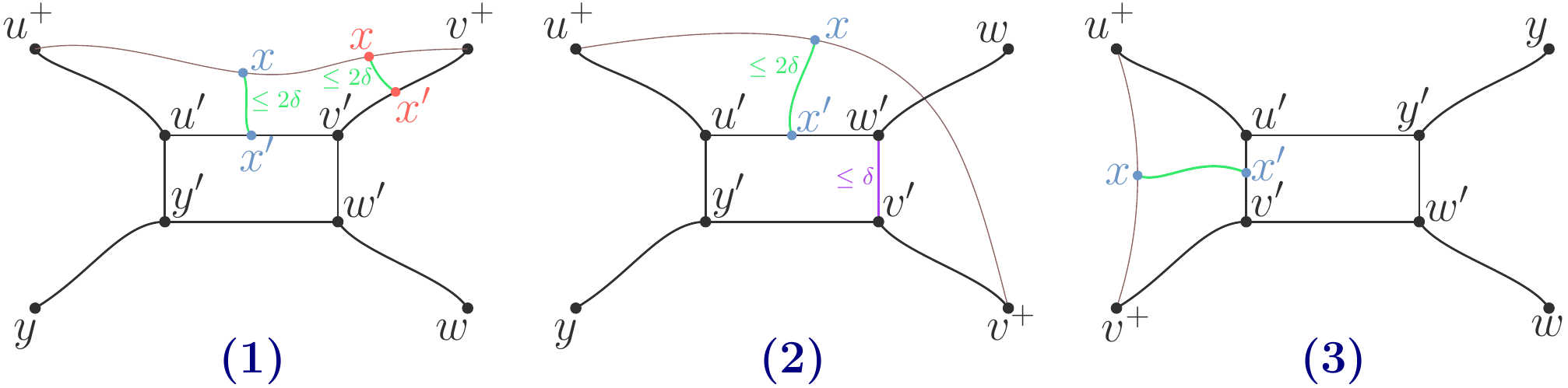}
\caption{
        \label{fig:LSCS_delta-hyperbolic}
         Cases 1-3 of Proposition \ref{prop:LSCS_delta-hyperbolic}. 
    }
\end{figure}
\section{Perspectives}
A first direction of interest would be to investigate the optimality of the sizes of our proper labeled sample compression schemes for the family of all balls in trees of cycles, interval graphs, and cube-free median graphs. Surely there must exist proper labeled sample compression schemes for cube-free median graphs of size smaller than $22$, but we wonder if such schemes of smaller size also exist for trees of cycles and interval graphs.
Another direction is to design proper sample compression schemes for balls of radius $r$ in trees of cycles or cube-free median graphs. Designing sample compression schemes of size $O(d)$ for balls in general median graphs $G$ of dimension $d$
is also open, as well as whether the VC-dimension of $\B(G)$ is  $O(d)$ or not. 
For general median graphs, it no longer holds that the interval between a diametral pair of $X^+$ contains a center of a ball
realizing 
$X$. However, one can show that $X^+$ contains $2d$ vertices whose convex hull
contains such a center. This convex hull can be $d$-dimensional and it is unclear how to encode the center in this region.
%

Other open questions are to design proper sample compression schemes of constant size for balls of planar graphs and of size $O(\omega(G))$ for balls of a chordal graph $G$. We showed that the former is possible for balls of radius~1 in Section~\ref{sec:planar}, and that the latter is possible for split graphs in Section~\ref{sec:split}.
Finding proper sample compression schemes of
size $O(\omega(G))$ for $\B(G)$  is also interesting for other classes of graphs from metric graph theory:
bridged graphs (generalizing chordal graphs) 
and Helly graphs; for their definitions and characterizations,
see 
\cite{BaCh_survey}.

\section*{Acknowledgements}
We are grateful to the anonymous referee for their careful reading and useful comments. This work has been supported by the ANR project DISTANCIA (ANR-17-CE40-0015), the European Research Council (ERC) consolidator grant No.~725978 SYSTEMATICGRAPH, and the Austrian Science Foundation (FWF, project Y1329).

\bibliographystyle{amsalpha}

\end{document}